\documentclass[10pt]{article}
\usepackage[utf8]{inputenc}
\usepackage{amsmath}
\usepackage{amsfonts}
\usepackage{amssymb}
\usepackage{cite}
\usepackage{textcomp}
\usepackage[nottoc, notlof, notlot]{tocbibind}
\usepackage{bbm}
\usepackage{bm}
\usepackage{amsthm}
\usepackage{tikz}
\usetikzlibrary{shapes, arrows}
\usepackage{graphicx}
\usepackage[justification=centering]{subfig}
\usepackage{epstopdf}
\graphicspath{{./Figures/}}
\usepackage[margin=2.4cm]{geometry}
\usepackage[]{algorithm2e}
\usepackage{array}
\newcolumntype{C}[1]{>{\centering\arraybackslash}p{#1}}

\DeclareMathOperator*{\argmin}{arg\,min}
\newcommand{\R}{\mathbb{R}}
\newcommand{\C}{\mathbb{C}}
\newcommand{\Z}{\mathbb{Z}}
\newcommand{\N}{\mathbb{N}}
\newcommand{\ie}{\emph{i.e.}, }
\newcommand{\eg}{\emph{e.g.}, }
\newcommand{\etal}{\emph{et al.} }
\def\V#1{{\ifcat#1\relax \boldsymbol{#1} \else \bf{#1} \fi} }   
\def\Spc#1{{\mathcal{#1}}}  
\def\M#1{{\bf{#1}}}  
\def\Op#1{{\mathrm{#1}}}  
\def\ee{\mathrm{e}} 
\def\jj{\mathrm{j}}

\newcommand{\MLone}{{\Spc M_{\Op L_1, \V \phi_0}(\R)}}
\newcommand{\HLtwo}{{\Spc H_{\Op L_2}(\R)}}

\newcommand{\opt}[1]{#1^\ast} 
\newcommand{\sopt}{\opt{s}}

\newcommand{\tik}{h}
\newcommand{\autocorr}{\rho}
\newcommand{\sqrta}{g}
\newcommand{\bseq}{b}
\newcommand{\sqrtb}{b^{1/2}}
\newcommand{\dis}{\mathrm{d}}

\newtheorem{theorem}{Theorem}
\newtheorem{definition}{Definition}
\newtheorem{proposition}{Proposition}

\newtheorem{remark}{Remark}

\title{Continuous-Domain Formulation of Inverse Problems for Composite Sparse-Plus-Smooth Signals}

\author{Thomas~Debarre,
			Shayan~Aziznejad,
        and~Michael~Unser,
\thanks{The authors are with the Biomedical Imaging Group, École polytechnique fédérale de Lausanne, 1015 Lausanne,
Switzerland (e-mail: thomas.debarre@gmail.com; shayan.aziznejad@epfl.ch; michael.unser@epfl.ch).}
\thanks{This research was supported by the European Research Council (ERC) under the European Union’s Horizon 2020 research and innovation programme, Grant 692726-GlobalBioIm, and by the Swiss National Science Foundation, Grant 200020\_184646/1.}
}
\begin{document}

\maketitle
\begin{abstract}
We present a novel framework for the reconstruction of 1D composite signals assumed to be a mixture of two additive components, one sparse and the other smooth, given a finite number of linear measurements. We formulate the reconstruction problem as a continuous-domain regularized inverse problem with multiple penalties. We prove that these penalties induce reconstructed signals that indeed take the desired form of the sum of a sparse and a smooth component. We then discretize this problem using Riesz bases, which yields a discrete problem that can be solved by standard algorithms. Our discretization is exact in the sense that we are solving the continuous-domain problem over the search space specified by our bases without any discretization error. We propose a complete algorithmic pipeline and demonstrate its feasibility on simulated data.
\end{abstract}

\section{Introduction}

In the traditional discrete formalism of linear inverse problems, the goal is to recover a signal $\V{c}_0 \in \R^N$ based on some measurement vector $\V{y} \in \R^M$. These measurements are typically acquired via a linear operator $\M{H} \in \R^{ M\times N }$ that models the physics of our acquisition system (forward model), so that $\M{H} \V{c}_0 \approx \V{y}$. The recovery is often achieved by solving an optimization problem that aims at minimizing the discrepancy between the measurements $\M{H} \V{c}$ of the reconstructed signal $\V{c}$ and the acquired data $\V{y}$. This data fidelity is measured with a suitable convex loss function $E:\R^M \times \R^M \to \R$, the prototypical example
being the quadratic error $E(\V x, \V y) =\frac12 \Vert \V x - \V y \Vert^2_2$. A regularization term is often added to the cost functional, which yields the optimization problem
\begin{align}
\label{eq:generaL_iiscrete_pb_intro}
\argmin_{\V{c} \in \R^N} \left\{ \underbrace{ E (\M{H} \V{c}, \V{y} )}_{\text{Data fidelity}} + \underbrace{\lambda \Spc{R} (\M{L} \V{c})}_{\text{Regularization}} \right\},
\end{align}
where $\Spc{R}$ is the regularization functional, $\M{L}$ specifies a suitable transform domain, and $\lambda > 0$ is a tuning parameter that determines the strength of the regularization. The use of regularization can have multiple motivations:
\begin{enumerate}
\item to handle the ill-posedness of the inverse problem, which occurs when different signals yield identical measurements;
\item to favor certain types of  reconstructed signal (\eg sparse or smooth) based on our prior knowledge;
\item to improve the conditioning of the inverse problem and thus increase its numerical stability and robustness to noise.
\end{enumerate}
Historically, the first instance of regularization dates back to Tikhonov~\cite{tikhonov1963solution} with a quadratic regularization functional $\Spc{R} = \Vert \cdot \Vert_2^2$. Tikhonov regularization constrains the energy of $\M{L} \V{c}$ which, when $\M L$ is a finite-difference matrix, leads to a smooth signal $\V c$. Tikhonov regularization has the practical advantage of being mathematically tractable which leads to a closed-form solution.
More recently, there has been growing interest in $\ell_1$ regularization $\Spc{R} = \Vert \cdot \Vert_1$, which has peaked in popularity for compressed sensing (CS)~\cite{donoho2006compressed, candes2006compressive, eldar2012compressed, foucart2013mathematical}. With $\ell_1$ regularization, the prior assumption is that the transform signal $\M L \V{c}_0$ is sparse, meaning that it has few nonzero coefficients: indeed, the $\ell_1$ norm can be seen as a convex relaxation of the $\ell_0$ ``norm", which counts the number of nonzero entries of a vector. The sparsity-promoting effect of $\ell_1$ regularization is well understood and documented~\cite{tibshirani1996regression, candes2006stable, unser2016representer}. It is now generally considered to be superior to Tikhonov regularization for most applications~\cite{hastie2015statistical}. Moreover, despite its non-differentiability, numerous efficient proximal algorithms based on the proximity operator of the $\ell_1$ norm have emerged to solve $\ell_1$-regularized problems~\cite{beck2009fast, beck2009fasta, chambolle2010first, boyd2011distributed}.

\subsection{Discrete Inverse Problems for Composite Signals}
Despite their success, $\ell_1$ and $\ell_2$ regularization methods are too simple to model many real-world signals.
In this paper, we investigate composite models of the form $s = s_1 + s_2$ where the two components have different characteristics. More precisely, $s_1$ is assumed to be sparse in some given domain and is treated with $\ell_1$ regularization, while $s_2$ is assumed to be smooth 
and is treated with $\ell_2$ regularization.
In discrete settings, a natural way of reconstructing such signals is to solve the optimization problem
\begin{align}
\label{eq:discrete_pb_intro}
\min_{\V{c}_1, \V{c}_2 \in \R^N} \left\{ E(\M{H} (\V{c}_1 +\V{c}_2 ), \V{y}) + \lambda_1 \Vert \M{L}_1 \V{c}_1 \Vert_1 + \lambda_2 \Vert \M{L}_2 \V{c}_2 \Vert_2^2 \right\},
\end{align}
where $\V c_1, \V c_2$ are the two components of the signal $\V c = \V c_1 + \V c_2$, $\lambda_1, \lambda_2 > 0$ are tuning parameters, $\M{L}_1 \in \R^{N \times N}$ is a sparsifying transform for $\V c_1$, and $\M{L}_2 \in \R^{N \times N}$ is a low-energy-promoting transform for $\V c_2$. Amongst others, this modeling is considered in~\cite{demol2004inverse, gholami2013balanced, naumova2014minimization, daubechies2016sparsity, grasmair2018adaptive, debarnot2021learning}.

\subsection{Continuous-Domain Formulation}

Until now, we have focused on the discrete setting, as it constitutes the vast majority of the inverse-problem literature for the obvious reason of computational feasibility. However, most real-world signals are inherently continuous.
Therefore, when feasible, to formulate the inverse problem in the continuous domain is a natural and desirable objective.

In this work, we adapt the discrete approach of~\eqref{eq:discrete_pb_intro} to 1D continuous-domain composite signals by solving an optimization problem of the form
\begin{align}
\label{eq:continuous_pb_intro}
\min_{s_1, s_2} \left\{ E(\V{\nu}({s}_1+{s}_2), \V{y} ) +  \lambda_1 \|\Op{L}_1\{ s_1\} \|_\Spc{M} + \lambda_2 \Vert \Op{L}_2\{s_2\} \Vert_{L_2}^2 \right\},
\end{align}
where $s_1, s_2$ are the two components of the signal $s = s_1 + s_2 : \R \to \R$, $\V \nu = (\nu_1, \ldots, \nu_M) : s \mapsto \V \nu (s) \in \R^M$ is the continuous-domain linear forward model, and $\Op L_1$, $\Op L_2$ are suitable continuously defined regularization operators. Typical choices are $\Op L_i = \Op D^{N_{0, i}}$ for $i \in \{1, 2\}$, where $\Op D$ is the derivative operator and $N_{0, i}$ the order of the derivative. The regularization norm $\Vert \cdot \Vert_{\Spc M}$ is the total-variation (TV) norm for measures, which is the continuous counterpart of the discrete $\ell_1$ norm \cite{candes2014towards, unser2017splines, aziznejad2018multi}.
We refer to this term as the generalized TV (gTV) regularizer due to the presence of the operator $\Op L_1$. Finally, $\Vert \cdot \Vert_{L_2}$ is the usual norm over the space $L_2(\R)$ of signals with finite energy; we refer to the corresponding term as the \emph{generalized Tikhonov} (gTikhonov) regularizer, which promotes smoothness in combination with the operator $\Op L_2$.

\subsection{Representer Theorems and Discretization}
A classical way of discretizing a continuous-domain problem is to reformulate it as a finite-dimensional one by relying on a \emph{representer theorem} that gives a parametric form of the solution. Prominent examples include representer theorems for problems formulated over reproducing-kernel Hilbert spaces (RKHS), which are foundational to the field of machine learning~\cite{wahba1990spline, schoelkopf2001generalized}. As demonstrated in~\cite[Theorem 3]{gupta2018continuous}, the minimization problem~\eqref{eq:continuous_pb_intro} over the component $s_2$ (with a fixed $s_1$) --- \ie gTikhonov regularization --- falls into this category: the representer theorem states that there is a unique solution of the form
\begin{align}
\label{eq:gTikh_sol_intro}
\sopt_2(x)  = p_2(x) + \sum_{m=1}^M a_{m,2} \tik_m(x),
\end{align}
where the additional component $p_2$ lies in the null space of $\Op L_2$ (\ie $\Op L_2 \{p_2 \} = 0$), $\tik_m$ is a (typically quite smooth) kernel function that is fully determined by the choice of $\nu_m$ and $\Op L_2$, and $a_{m, 2} \in \R$ are expansion coefficients. Therefore, to solve the continuous-domain problem, one need only optimize over the $a_{m, 2}$ coefficients and the null-space component $p_2$ which lives in a finite-dimensional space. This leads to a standard finite-dimensional problem with Tikhonov regularization.

Concerning the minimization over the component $s_1$ (gTV regularization), several representer theorems give a parametric form of a sparse solution in different settings~\cite{fisher1975spline, unser2017splines, boyer2019representer, bredies2019sparsity, fageot2020tv}. The case of our exact setting is tackled by~\cite[Theorem 4]{gupta2018continuous}, which states that there is an $\Op L_1$-spline solution of the form
\begin{align}
\label{eq:gTV_sol_intro}
\sopt_1(x)  = p_1(x) + \sum_{k=1}^K a_{k,1} \rho_{\Op L_1}(x - x_k),
\end{align}
where $a_{k,1}, x_k \in \R$, $\rho_{\Op L_1}$ is a Green's function of $\Op L_1$ (\ie $\Op L_1 \{ \rho_{\Op L_1} \} = \delta$, where $\delta$ is the Dirac impulse), $K$ is the number of atoms of $s_1$ which is bounded by $K \leq (M-N_{0,1})$, $N_{0,1}$ being the dimension of the null space of $\Op L_1$, and $p_1$ lies in the null space of $\Op L_1$. For example, when $\Op L_1 = \Op D^{N_{0, 1}}$, $s_1$ is a piecewise polynomial of degree $(N_{0,1} - 1)$ with smooth junctions at the knots $x_k$. These representer theorems have paved the way for various exact discretization methods. In the gTikhonov case, one can optimize over the $a_{m,2}$ coefficients in~\eqref{eq:gTikh_sol_intro} directly~\cite{gupta2018continuous}. For the gTV case~\eqref{eq:gTV_sol_intro}, grid-based techniques using a well-conditioned B-spline basis~\cite{debarre2019b} as well as grid-free techniques~\cite{flinth2019exact} have been proposed.

\subsection{Our Contribution}
In this work, we show that the representer theorems presented in the previous section can be combined into a composite one when dealing with Problem~\eqref{eq:continuous_pb_intro}. More specifically, we prove that there exists a solution to~\eqref{eq:continuous_pb_intro} of the form $\sopt_1 = \sopt_1 + \sopt_2$ such that $\sopt_1$ is of the form~\eqref{eq:gTV_sol_intro} and $\sopt_2$ is of the form~\eqref{eq:gTikh_sol_intro}: a ``sparse plus smooth'' solution.
Building on this representation, we propose an exact discretization scheme. Both components $s_i$ for $i \in \{ 1, 2\}$ are expressed in a suitable Riesz basis as $s_i = \sum_k c_i[k] \varphi_{i, k}$, where $c_i[k]$ are the coefficients to be optimized. This leads to an infinite-dimensional optimization problem reminiscent of the infinite-dimensional compressed sensing framework of Adcock and Hansen~\cite{adcock2015generalized}.

To solve this infinite-dimensional problem numerically, we cast it as a finite-dimensional problem under some mild assumptions. This requires a careful handling of the boundaries of our interval of interest. In our implementation, we choose basis functions $\varphi_{1, k} = \beta_{\Op L_1} ( \cdot - k)$ and $\varphi_{2, k} = \beta_{\Op L_2^\ast \Op L_2}(\cdot - k)$, where $\beta_{\Op L}$ is the B-spline for the operator $\Op L$. B-splines are popular choices of basis functions \cite{boor2001practical, unser1993b, unser1999splines}, in large part due to their minimal-support property. Indeed, $\beta_{\Op L_i}$ has finite support when $\Op L_i = \Op D^{N_0}$, and it is the shortest-support generating function of the space of uniform $\Op L_i$ splines \cite{schoenberg1973cardinal}. We show that optimizing over the spline coefficients leads to a discrete problem similar to~\eqref{eq:discrete_pb_intro} of the form
\begin{align}
\label{eq:discretized_pb_intro}
\min_{(\V{c}_1, \V{c}_2) \in \R^{N_1} \times \R^{N_2} } \left\{ E(\M{H}_1\V{c}_1 + \M{H}_2\V{c}_2 , \V{y}) + \lambda_1 \Vert \M{L}_1 \V{c}_1 \Vert_1 + \lambda_2 \Vert \M{L}_2 \V{c}_2 \Vert_2^2 \right\},
\end{align}
where $\M H_i \in \R^{M \times N_i}$ and $\M L_i \in \R^{P_i \times N_i}$ for $i\in \{1, 2\}$. This discretization is exact in the sense that it is equivalent to the continuous problem~\eqref{eq:continuous_pb_intro} when each component $s_i$ lies in the space generated by the basis functions $\{\varphi_{i, k} \}_{k \in \Z}$. This is a consequence of our informed choice of these basis functions $\varphi_{i, k}$. Moreover, the short support of the B-splines leads to well-conditioned $\M H_i$ matrices and, thus, to a computationally feasible problem.

\subsection{Related Works}
The use of multiple regularization penalties is quite common in the litterature. However, in most cases, each penalty is applied to the full signal instead of a component-wise \cite{belge2002efficient, roth2005fields, chen2008multi, lu2010multi, wang2012multi, abhishake2016multi}. A prominent example of such an approach is the elastic net \cite{zou2005regularization}, which is widely used in statistics. The spirit of these approaches is however quite different from ours: the reconstructed signal is encouraged to satisfy different priors \emph{simultaneously}. Conversely, in~\eqref{eq:discrete_pb_intro}, each component satisfies different priors independently from the others, which will give very different results.

The model of Meyer \cite{meyer2001oscillating} and its generalization by Vese and Osher \cite{vese2003modeling, vese2004image} follow the same idea as Problem~\eqref{eq:continuous_pb_intro}, with the important difference that they use calculus-of-variation techniques to solve it. There is a connection as well with the Mumford-Shah functional \cite{mumford1989optimal}, which is commonly used to segment an image in piecewise-smooth regions. The main difference lies in the fact that the optimization is not performed over the different components of the signal, but over the region boundaries. Another difference is that these models assume that one has full access to the noisy signal over a continuum, whereas~\eqref{eq:continuous_pb_intro} assumes that we only have access to some discrete measurements specified by the forward model $\V \nu$.

\subsection{Outline}
In Section~\ref{sec:preliminaries}, we give the necessary mathematical preliminaries to formulate our optimization problem. In Section~\ref{sec:continuous_pb}, we formulate the continuous-domain problem and present our representer theorem, which is our main theoretical result. In Section~\ref{sec:discretization}, we explain our discretization strategy, which relies on the selection of suitable Riesz bases. Finally, in Section~\ref{sec:applications}, we present experiments on simulated data.

\section{Preliminaries}
\label{sec:preliminaries}
\subsection{Operators and Splines}
The crucial elements of our formulation are the regularization operators $\Op{L}_1$ and $\Op{L}_2$. In this section, we specify which type of operators are suitable in our framework.

Let $\Spc{S}'(\R)$ denote the space of tempered distributions, defined as the dual of the Schwartz space $\Spc{S}(\R)$ of infinitely smooth functions on $\R$ whose successive derivatives are rapidly decaying. Let $\Spc{F}$ be the generalized Fourier transform with $\widehat{f} \triangleq \Spc{F} \lbrace f \rbrace$. Then, it is a standard result in distribution theory that the frequency response $\widehat{L} = \mathcal{F}\{\Op L\{\delta\}\}$ of an ordinary differential operator $\Op{L} : \Spc{S}'(\R) \to \Spc{S}'(\R)$ is a slowly increasing smooth function $\widehat{L}: \R \to \C$ ~\cite[Chapter 7, §5]{schwartz1951theorie}. Moreover, for any $f \in \Spc{S}'(\R)$, we have that $\Spc{F} \lbrace \Op{L} \lbrace f \rbrace \rbrace = \widehat{L} \widehat{f} \in \Spc S'(\R)$.
Next, we require $\Op{L}$ to be spline-admissible in the sense of Definition~\ref{def:spline_admissible}.
\begin{definition}[Spline-admissible operator]
\label{def:spline_admissible}
A continuous LSI operator $\Op{L} : \Spc{S}'(\R) \to \Spc{S}'(\R)$ is \emph{spline-admissible} if it verifies the following properties:
\begin{itemize}
\item there exists a function of slow growth $\rho_{\Op{L}} : \R \to \R$ (the Green's function of $\Op{L}$) that satisfies $\Op{L} \lbrace \rho_{\Op{L}} \rbrace = \delta$;
\item its null space $\Spc{N}_{\Op{L}}  = \{ f \in \Spc{S}'(\R) :  \Op{L} \lbrace f \rbrace = 0 \}$ has finite dimension $N_0$.
\end{itemize}
\end{definition}

The prototypical example of a spline-admissible operator is the multiple-order derivative $\Op{L} = \Op{D}^{N_0}$ for $N_0 \geq 1$. Its causal Green's function is the one-sided power function $\rho_{\Op{L}} = \frac{x_+^{N_0-1}}{(N_0-1)!}$, where $x_+ = \max(0, x)$. The null space of $\Op{L}$ is the space of polynomials of degree less than $N_0$.

A spline-admissible operator $\Op{L}$ specifies the family of $\Op{L}$-splines provided in Definition~\ref{def:non_uniform_spline}.
\begin{definition}[Nonuniform $\Op{L}$-spline]
\label{def:non_uniform_spline}
Let $\Op{L}$ be a spline-admissible operator in the sense of Definition~\ref{def:spline_admissible}. A nonuniform $\Op{L}$-spline with $K$ knots $x_1 < \cdots < x_K$ is a function $s : \R \mapsto\R$ that verifies
\begin{align}
\label{eq:spline_def_innovation}
\Op{L} \lbrace s \rbrace(x) = \sum_{k=1}^K a_k \delta (x - x_k),
\end{align}
where $a_k \in \R$ is the amplitude of the $k$th singularity. The weighted sum of Dirac impulses in \eqref{eq:spline_def_innovation} is known as the \emph{innovation} of $s$. The spline $s$ can equivalently be written as
\begin{align}
\label{eq:spline_def_green}
s(x) = p(x) + \sum_{k=1}^K a_k \rho_{\Op{L}}(x-x_k),
\end{align}
where $p \in \Spc{N}_{\Op{L}}$. 
\end{definition}
For example, the operator $\Op{L} = \Op{D}^{N_0}$ leads to the well-known polynomial splines, which are piecewise polyomials of degree $(N_0-1)$ and of differentiability class $\Spc{C}^{N_0-2}$.


\subsection{Native Spaces}

\subsubsection{Sparse Component}

The other crucial elements of our framework are the native spaces for each component. Let $\Spc{M}(\R)$ be the space of bounded Radon measures, which is known by the Riesz-Markov theorem~\cite[Chapter 6]{rudin1986real} to be the continuous dual of $C_0(\R)$. The latter is the space of continuous functions vanishing at infinity, which is a Banach space when equipped with the supremum norm $\Vert \cdot \Vert_\infty$. The sparsity-promoting regularization norm $\Vert \cdot \Vert_\Spc{M}$ is defined for a tempered distribution $w \in \Spc{S}'(\R)$ as
\begin{align}
\Vert w \Vert_\Spc{M} \triangleq \sup_{\varphi \in \Spc{S}(\R), \Vert \varphi \Vert_\infty = 1} \langle w, \varphi \rangle.
\end{align}
Practically, the two critical features of the $\Vert \cdot \Vert_\Spc{M}$ norm are the following:
\begin{enumerate}
\item it generalizes the $L_1$ norm in the sense that $\Vert w \Vert_\Spc{M} = \Vert w \Vert_{L_1}$ for any $w \in L_1(\R)$;
\item the $\Vert \cdot \Vert_\Spc{M}$ norm of a weighted sum of Dirac impulses is $\Vert \sum_{k} a_k \delta(\cdot - x_k) \Vert_\Spc{M} = \sum_k \vert a_k \vert$.
\end{enumerate}

Accordingly, the native space for $s_1$ in~\eqref{eq:continuous_pb_intro} is defined as
\begin{align}
\label{eq:native_space_1}
\Spc{M}_{\Op{L}_1}(\R) = \{ s \in \Spc{S}'(\R): \Op{L}_1 \{ s\} \in \Spc{M}(\R) \},
\end{align}
which is a Banach space when equipped with the direct-sum topology. It is also the largest space for which the regularization is well-defined. We refer to~\cite{unser2019native} for technical details on the construction of $\Spc{M}_{\Op{L}_1}(\R)$.

\subsubsection{Smooth Component}
The regularization norm $\Vert \cdot \Vert_{L_2}$ for the smooth component $s_2$ in~\eqref{eq:continuous_pb_intro} is defined over the Hilbert space $L_2(\R)$. The corresponding native space of the smooth component $s_2$ is the Hilbert space
\begin{align}
\label{eq:native_space_2}
\Spc{H}_{\Op{L}_2}(\R) = \{ f \in \Spc{S}'(\R): \Op{L}_2 \{ f\} \in L_2(\R) \}.
\end{align}

\subsubsection{Boundary Conditions}
Finally, to ensure the well-posedness of our optimization problem, boundary conditions need to be introduced for one of the two native spaces. Let $\Spc{N}_0 = \Spc{N}_{\Op{L}_1} \cap  \Spc{N}_{\Op{L}_2}$ be the intersection of the null spaces. We introduce a biorthogonal system $(\V{\phi}_0, \V{p}_0)$ for $\Spc{N}_0$ in the sense of~\cite[Definition 3]{unser2017splines}. An example of a valid choice is given in Appendix~\ref{app:boundary_conditions}. The search space with boundary conditions $\V\phi_0$ is then given by
\begin{align}
\label{eq:native_space_1_restricted}
\Spc{M}_{\Op{L}_1, \V\phi_0}(\R) = \{ f\in \Spc{M}_{\Op{L}_1}(\R) : \V\phi_0(f) = \V{0} \}.
\end{align}

\section{Continuous-Domain Inverse Problem}
\label{sec:continuous_pb}

Now that the relevant spaces have been introduced, we present in Theorem~\ref{thm:continuous_RT} the optimization task that we use to reconstruct sparse-plus-smooth composite signals. This representer Theorem gives a parametric form of a solution of our optimization problem; the proof is given in Appendix~\ref{app:proof_RT}.

\begin{theorem}[Continuous-domain representer theorem]
\label{thm:continuous_RT}
Let $\Op{E} : \R^M \times \R^M \to \R$ be a nonnegative, coercive, proper,  convex, and lower-semicontinuous functional. Let $\Op{L}_1, \Op{L}_2$ be spline-admissible operators in the sense of Definition~\ref{def:spline_admissible} and let $\V\nu = (\nu_1, \hdots, \nu_M)$ be a linear measurement operator composed of the $M$ linear functionals $\nu_m: f \mapsto \nu_m(f) \in \R$ that are $\text{weak}^\ast$-continuous over $\Spc{M}_{\Op{L}_1} (\R)$ and over $\Spc{H}_{\Op{L}_2}(\R)$. We assume that $\Spc{N}_{\V\nu} \cap (\Spc{N}_{\Op{L}_1} + \Spc{N}_{\Op{L}_2}) = \{0 \}$, where $\Spc{N}_{\V\nu}$ is the null space of $\V\nu$ (well-posedness assumption). Then, for any $\lambda_1, \lambda_2 > 0$, the optimization problem
\begin{align}
\label{eq:continuous_pb}
 \Spc{S} =& \left\{ \argmin_{\substack{s_1 \in \MLone \\ s_2 \in \HLtwo}}  \Spc{J}(s_1, s_2) \right\} \qquad \text{with} \nonumber \\
\Spc{J}(s_1, s_2) =& \Op{E}(\V{\nu}({s}_1+{s}_2),\V{y}) +  \lambda_1 \|\Op{L}_1\{ s_1\} \|_\Spc{M}+ \lambda_2 \Vert \Op{L}_2\{s_2\} \Vert_{L_2}^2 
\end{align}
has a solution $(\sopt_1, \sopt_2) \in \Spc{S}$ with the following components:
\begin{itemize}
    \item {the component $\sopt_1$ is a nonuniform $\Op{L}_1$-spline of the form
\begin{align}
\label{eq:continuous_sol_s1}
\sopt_1(x) = p_1(x) + \sum_{k=1}^{K_1} a_{1,k} \rho_{\Op{L}_1}( x - x_{k} )
\end{align}
for some $K_1 \leq (M - N_{0,1})$, where $p_1 \in \Spc{N}_{\Op{L}_1} $, and $a_{1,k}, x_{k} \in \R$;
    }
    \item {
     the component $\sopt_2$ is of the form
\begin{align}
\label{eq:continuous_sol_s2}
\sopt_2(x)  = p_2(x) + \sum_{m=1}^M a_{2, m} \tik_m(x),
\end{align}
where $\tik_m(x) = \left( \nu_m \ast \Spc{F}^{-1} \left\{ \frac{1}{\vert \widehat{L}_2 \vert ^2} \right\} \right) (x)$, $p_2 \in \Spc{N}_{\Op{L}_2}$,  $a_{2,k} \in \R$, and where $\sum_{m=1}^M a_{2, m} \langle q_2, \nu_m \rangle = 0$ for any $q_2 \in \Spc{N}_{\Op{L}_2}$.
    }
\end{itemize}
 Moreover, for any pair of solutions $(\sopt_1, \sopt_2), (\tilde{s}^\ast_1, \tilde{s}^\ast_2) \in \Spc S$, $\sopt_2$ and $\tilde{s}^\ast_2$ differ only up to an element of the null space $\Spc{N}_{\Op{L}_2}$, so that $(\sopt_2 - \tilde{s}^\ast_2) \in \Spc{N}_{\Op{L}_2}$.
\end{theorem}

A pleasing outcome of Theorem~\ref{thm:continuous_RT} is that it combines Theorems 3 and 4 of~\cite{gupta2018continuous} into one. There is, however, an added technicality due to the boundary conditions $\V\phi_0$. The latter are necessary to ensure the well-posedness of Problem~\eqref{eq:continuous_pb}. Otherwise, for any $(\sopt_1, \sopt_2) \in \Spc{S}$ and $p \in \Spc{N}_0$, we would have that
$(\sopt_1 + p, \sopt_2 - p) \in \Spc{S}$ which would imply that $\Spc{S}$ is unbounded. Note, however, that these conditions do not restrict the search space, since $\Spc{M}_{\Op{L}_1, \V{\phi}_0} (\R) + \Spc{H}_{\Op{L}_2}(\R) = \Spc{M}_{\Op{L}_1} (\R) + \Spc{H}_{\Op{L}_2}(\R)$.

\section{Exact Discretization}
\label{sec:discretization}


In order to discretize Problem~\eqref{eq:continuous_pb}, we restrict the search spaces $\MLone$ and $\HLtwo$. The standard approach to achieve this is to choose a sequence of appropriate basis functions $\{ \varphi_{i, k} \}_{k \in \Z}$ that span the reconstruction spaces
\begin{align}
\label{eq:search_space}
V_i(\R) = \left\{\sum_{k\in \Z} c_i[k] \varphi_{i,k}:  c_i \in V_i(\Z) \right\}
\end{align}
for $i \in \{1, 2\}$ that are subject to the constraints $V_1(\R) \subset \MLone$ and $V_2(\R) \subset \HLtwo$. These continuous spaces are linked to discrete spaces $V_i(\Z)$, the choices of which will be made explicit in \eqref{eq:native_space_sequence1} and \eqref{eq:native_space_sequence2}. More precisely, there is a one-to-one mapping between them using the basis functions $\varphi_{i,k}$.



\subsection{Riesz Bases and B-Splines}
For numerical purposes, a desirable property is that our basis functions satisfy the Riesz property. Riesz bases are highly important concepts in that they generalize orthonormal bases, while leaving more flexibility for other desirable properties such as short support \cite{daubechies1992ten}.
\begin{definition}[Riesz basis]
\label{def:riesz}
A sequence of functions $\{ \varphi_k \}_{k \in \Z}$ with $\varphi_k \in L_2(\R)$ is said to be a Riesz basis if there exist constants $0 < A \leq B$ such that, for any $c \in \ell_2(\Z)$, we have that
\begin{align}
\label{eq:riesz}
A \Vert c \Vert_{\ell_2} \leq \left\Vert \sum_{k \in \Z} c[k] \varphi_{k} \right\Vert_{L_2} \leq B \Vert c \Vert_{\ell_2}.
\end{align}
\end{definition}

 Popular examples of Riesz bases are B-spline bases, which are introduced in Definition~\ref{def:B-spline}.
\begin{definition}[B-spline]
\label{def:B-spline}
The B-spline for a spline-admissible operator $\Op L$ is characterized by a finite-difference-like filter $(d_{\Op L}[k])_{k \in \Z}$, and is defined as 
\begin{align}
\label{eq:B-spline}
\beta_{\Op L} (x) = \Spc{F}^{-1} \left\{ \frac{\sum_{k \in \Z} d_{\Op L}[k] \ee^{-\jj k (\cdot )} }{ \widehat{L}(\cdot) } \right\}(x).
\end{align}
The criteria for choosing a valid filter $d_{\Op L}$ for a general class of operators $\Op L$ are given in \cite[Theorem 2.7]{amini2018universal}.
\end{definition}

The best-known example of a B-spline is the polynomial B-spline for the operator $\Op L = \Op D^{N_0}$, whose filter $d_{\Op L}$ is characterized by its $z$-transform $D_{\Op L}(z) = (1 - z^{-1})^{N_0}$. The corresponding B-spline $\beta_{\Op L}$ is supported in $[0, N_0]$.

A key feature of B-splines is that they are the $\Op L$-splines with the shortest support or, when finite support is impossible, with the fastest decay. 
Moreover, by \cite[Theorem 2.7]{amini2018universal}, for a valid B-spline $\beta_{\Op L}$ as specified by Definition~\ref{def:B-spline}, the sequence of functions $\{ \beta_{\Op L} (\cdot - k) \}_{k \in \Z}$ forms a Riesz basis in the sense of Definition~\ref{def:riesz}.



It is clear from~\eqref{eq:B-spline} that the innovation of the B-spline is a sum of Dirac impulses given by
\begin{align}
\label{eq:innovation_B-spline}
\Op L \{ \beta_{\Op L} \} = \sum_{k \in \Z} d_{\Op L}[k] \delta( \cdot - k).
\end{align}

\subsection{Choice of Basis Functions}
\label{sec:basis_functions}
We now present and discuss our choice for the basis functions $\varphi_{1, k}$ and $\varphi_{2, k}$.

\subsubsection{Sparse Component}
For the sparse component, we choose basis functions $\varphi_{1, k} = \beta_{\Op L_1}(\cdot - k)$ (defined in~\eqref{eq:B-spline}) for all $k \in \Z$.
With this choice,
\begin{align*}
V_1(\R) = \left\{ f = \sum_{k \in \Z} c_1[k] \varphi_{1, k}: c_1 \in V_1(\Z) \right\} \subset \Spc M_{\Op L_1, \phi_0}(\R)
\end{align*}
with the digital-filter space
\begin{align}
\label{eq:native_space_sequence1}
V_1(\Z) = \Bigg\{(c_1[k])_{k \in \Z} : \ &(d_{\Op L_1} \ast c_1) \in \ell_1(\Z) \text{ and }  \sum_{k \in \Z} c_1[k ]\V \phi_0(\varphi_{1, k}) = \V 0 \Bigg\},
\end{align}
is the largest possible native reconstruction space~\cite[Equation (22)]{debarre2019hybrid}.
The choice of the basis functions $\varphi_{1, k}$ is guided by the following considerations:
\begin{itemize}
\item they generate the space of uniform $\Op L_1$ splines. This conforms with Theorem~\ref{thm:continuous_RT}, which states that the component $\sopt_1$ is an $\Op L_1$-spline;
\item they enable exact computations in the continuous domain. In particular, we have that $\Vert \Op L_1 \{ \sum_{k \in \Z} c_1[k] \varphi_{1, k} \} \Vert_{\Spc M} = \Vert d_{\Op L_1}\ast c_1 \Vert_{\ell_1}$;
\item the Riesz-basis property of B-splines leads to a well-conditioned system matrix, which is paramount in numerical applications.
\end{itemize}
B-splines are the only functions that satisfy all these properties. Based on these criteria, B-splines are thus optimal.

\subsubsection{Smooth Component}

At first glance, the most natural choice for $\varphi_{2, k}$ is to select the basis functions suggested by~\eqref{eq:continuous_sol_s2} in Theorem~\ref{thm:continuous_RT}: $\tik_m$ for $1 \leq m \leq M$ and a basis of $\Spc N_{\Op L_2}$, which yield a finite number $M + N_{0,2}$ of basis functions. However, this approach runs into the following hitches:
\begin{itemize}
\item the basis functions $\tik_m$ are typically increasing at infinity, which contradicts the Riesz-basis requirement and leads to severely ill-conditioned optimization tasks~\cite{gupta2018continuous};
\item depending on the measurements operator $\V \nu$, $\tik_m$ may lack a closed-form expression.
\end{itemize}

We therefore focus on these criteria, in a spirit similar to~\cite{bohra2020computation}. The $\varphi_{2, k}$ are chosen to be regular shifts of a generating function $\varphi_2$, with $\varphi_{2, k} = \varphi_2( \cdot - k)$ such that $\left\{ \Op L_2 \{ \varphi_{2, k} \} \right\}_{k \in \Z}$ forms a Riesz basis in the sense of Definition~\ref{def:riesz}. Contrary to $\varphi_{1, k}$, these requirements allow for many choices of $\varphi_{2, k}$. In order to perform exact discretization, one then only needs to compute the following autocorrelation filter.
\begin{proposition}[Autocorrelation filter for the smooth component]
\label{prop:autocorrelation}
Let $\varphi_2$ be a generating function such that $\varphi_{2, k} = \varphi_2(\cdot - k)$ 
form a Riesz basis. Then, the following two items hold:
\begin{itemize}
\item the inner product $\langle \Op L_2 \{\varphi_{2, k} \}, \Op L_2 \{\varphi_{2, k'} \} \rangle_{L_2}$ only depends on the difference $(k - k')$. We can thus introduce the autocorrelation filter
\begin{align}
\label{eq:autocorrelation}
\autocorr [k] &= \langle \Op L_2 \{\varphi_{2, k} \}, \Op L_2 \{\varphi_{2, 0} \} \rangle_{L_2} \nonumber \\
&= \langle \Op L_2 \{\varphi_{2, k+k'} \}, \Op L_2 \{\varphi_{2, k'} \} \rangle_{L_2}
\end{align}
for any $k, k' \in \Z$;
\item the filter $\autocorr$ is positive semidefinite, with $\sum_{k, k' \in \Z} c[k] c[k'] \autocorr [k - k'] \geq 0$ for any finitely supported real digital filter $c$.
\end{itemize}

\end{proposition}
\begin{proof}
The first item is proved with a simple change of variable in the integral that defines the inner product. The second item is derived by observing that, for any $c_2$, we have
\begin{align}
\label{eq:positive_definite}
\nonumber
\left\Vert \Op L_2 \left\{ \sum_{k \in \Z} c_2[k] \varphi_{2, k} \right\} \right\Vert_{L_2}^2 &= \sum_{k, k' \in \Z} c_2[k]c_2[k'] \langle \Op L_2 \{ \varphi_{2, k} \}, \Op L_2 \{ \varphi_{2, k'} \} \rangle \\
&= \sum_{k, k' \in \Z} c_2[k]c_2[k'] \autocorr[ k - k'] \geq 0.
\end{align}
\end{proof}

\subsection{Formulation of the Discrete Problem}
The autocorrelation filter introduced in Proposition~\ref{prop:autocorrelation} enables us to discretize Problem~\eqref{eq:continuous_pb} in an exact way in the $V_i(\R)$ spaces.
\begin{proposition}[Riesz-basis Discretization]
\label{prop:riesz_discretization}
Let $\varphi_{i, k}$ be chosen as specified in Section~\ref{sec:basis_functions} for $i \in \{ 1, 2\}$, $k \in \Z$, and
\begin{align}
\label{eq:discrete_pb}
\Spc{S}_\dis = \left\{ \argmin_{(c_1, c_2) \in V_1(\Z) \times V_2(\Z)} \Spc{J}_\dis(c_1, c_2) \right\}.
\end{align}
The cost function is given by
\begin{align}
\label{eq:discrete_cost}
\nonumber
\Spc{J}_\dis(c_1, c_2) =& E \left( \sum_{k_1 \in \Z} c_1[k] \V \nu ( \varphi_{1, k} )  + \sum_{ k \in \Z} c_2[k] \V \nu ( \varphi_{2, k}), \V y \right) \\
& + \lambda_1 \Vert d_{\Op L_1} \ast c_1 \Vert_{\ell_1} + \lambda_2 \langle c_2, \autocorr \ast c_2 \rangle_{\ell_2},
\end{align}
where $d_{\Op L_1}$ is the finite-difference-like filter from Definition~\ref{def:B-spline}, $\autocorr$ is defined in Proposition~\ref{prop:autocorrelation}, and $\langle \cdot , \cdot \rangle_{\ell_2}$ is the inner product over $\ell_2(\Z)$.
Then, Problem~\eqref{eq:discrete_pb} is equivalent to a restriction of the search spaces $\MLone$ and $\HLtwo$ to the spaces $V_1(\R)$ and $V_2(\R)$ defined in~\eqref{eq:search_space}, respectively, so that
\begin{align}
\label{eq:restricted_continuous_pb}
 \Spc{S}_{\mathrm{res}} = \left\{ \argmin_{(s_1, s_2) \in V_1(\R) \times V_2(\R)}  \Spc{J}(s_1, s_2) \right\},
\end{align}
in the sense that there exists a bijective linear mapping $(c_1, c_2) \mapsto \left(\sum_{k \in \Z} c_1[k] \varphi_{1, k}, \sum_{k \in \Z} c_2[k] \varphi_{2, k} \right)$ from $\Spc S_\dis$ to $\Spc S_\mathrm{res}$.
\end{proposition}
\begin{proof}
By plugging the expansions $s_i = \sum_{k \in \Z} c_i[k] \varphi_{i, k}$ into the cost function $\Spc J$, using the linearity of $\V \nu$, we get the data-fidelity term of~\eqref{eq:discrete_cost}. Using~\eqref{eq:innovation_B-spline}, we readily deduce that $\Vert \Op L_1 \{ \sum_{k \in \Z} c_1[k] \varphi_{1, k} \} \Vert_{\Spc M} = \Vert d_{\Op L_1} \ast c_1 \Vert_{\ell_1}$ \cite[Equation (25)]{debarre2019b}. As for the second regularization term, we observe that

\begin{align}
\langle c_2, c_2 \ast \autocorr \rangle_{\ell_2}
&=  \sum_{k, k' \in \Z} c_2[k] c_2[k'] \autocorr [k - k'] \nonumber \\
&=  \left\Vert \Op L_2 \left\{ \sum_{k \in \Z} c_2[k] \varphi_{2, k} \right\} \right\Vert_{L_2}^2,
\end{align}
using~\eqref{eq:positive_definite} for the last step. This proves the equivalence between Problems~\eqref{eq:restricted_continuous_pb} and \eqref{eq:discrete_pb}, up to the specified mapping which is indeed a bijective linear mapping due to the Riesz-basis properties of $\{ \varphi_{1, k} \}_{k \in \Z}$ and $\{ \varphi_{2, k} \}_{k \in \Z}$.
\end{proof}

\section{Practical Implementation}
We now discuss how to solve our discretized problem~\eqref{eq:discrete_pb} in practice, which involves recasting it as a finite-dimensional problem.

\subsection{Finite Domain Assumptions }
\label{sec:assumptions_finite_pb}
To solve problem~\eqref{eq:discrete_pb} numerically in an exact way, we must make assumptions that will enable us to restrict the problem to a finite interval of interest.
\begin{enumerate}
\item The operators $\Op L_i$ for $i \in \{1, 2\}$ admit a B-spline with finite support, which implies that the filters $d_{\Op L_i}$ (introduced in Definition~\eqref{eq:B-spline}) and $\autocorr$ 
(Proposition~\ref{prop:autocorrelation_B-splines}) have finite support. Without loss of generality, the support of $d_{\Op L_i}$ is chosen to be $[0 \ldots D_i-1 ]$ (of length $D_i > 0$), which leads to causal B-splines $\beta_{\Op L_i}$.
\item The measurement functionals $\nu_m$ are supported in an interval $I_T = [0, T]$, where $T \in \N$.
\end{enumerate}

The first item is fulfilled for common one-dimensional operators $\Op L_1$ such as ordinary differential operators~\cite{unser2005cardinal} or rational operators~\cite{unser2005cardinalII}.

The second assumption is natural and is often fulfilled in practice, for instance in imaging with a finite field of view. The support length $T$ then roughly corresponds to the number of grid points in the interval of interest. Note that, for simplicity, we only consider integer grids. However, the finesse of the grid can be tuned at will by adjusting $T$ and rescaling the problem over the interval of interest.

\subsection{Choice of Basis Functions $\varphi_{2, k}$}
\label{sec:choice_phi2}

For our implementation, we make a specific choice of basis functions $\varphi_{2, k}$ for the second component, since many different choices satisfy the requirements of Section~\ref{sec:basis_functions}. We choose the $\Op L_2 ^\ast \Op L_2$ B-spline basis $\varphi_2 = \beta_{\Op L_2 ^\ast \Op L_2}$ and $\varphi_{2, k} = \varphi_2( \cdot - k)$, where $\Op L_2 ^\ast$ denotes the adjoint operator of $\Op L_2$. In addition to the items discussed in Section~\ref{sec:basis_functions}, this choice has the following advantages:
\begin{itemize}
\item the generator $\varphi_2$ has a simple explicit expression that does not depend on the measurement operator $\V \nu$;
\item the autocorrelation filter $\autocorr$ also has a simple expression, as will be shown in Proposition~\ref{prop:autocorrelation_B-splines};
\item in the special case of the sampling operator $\nu_m = \delta( \cdot - x_m)$, where the $x_m$ are the sampling locations, this choice conforms with~\eqref{eq:continuous_sol_s2} in Theorem~\ref{thm:continuous_RT} since $\sopt_2$ is then an $\Op L_2 ^\ast \Op L_2$-spline. Note, however, that we do not exploit the knowledge that $\sopt_2$ has knots at the sampling locations $x_m$.
\end{itemize}

\subsection{Formulation of the Finite-Dimensional Problem}

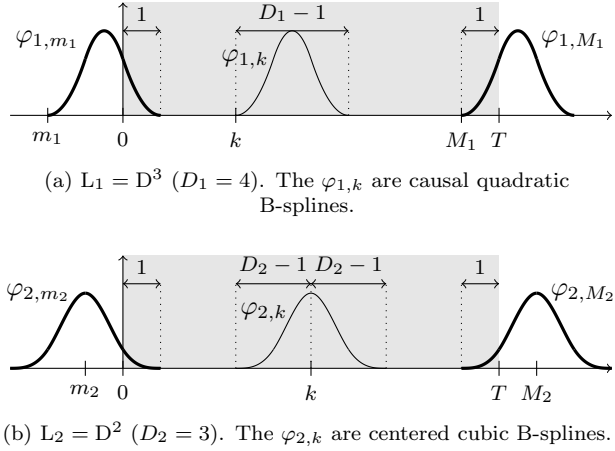
\begin{figure}
\centering
\subfloat[$\Op L_1 = \Op D^3$ ($D_1 = 4$). The $\varphi_{1,k}$ are causal quadratic B-splines.]{
\begin{tikzpicture}[scale=0.5]
  \path [fill=gray!20] (0,0) rectangle (10,3);
  \draw[->] (-3, 0) -- (13, 0);
  \draw[->] (0, 0) -- (0, 3);
   \draw (0, 0.2) -- (0, -0.2) node[below] {\footnotesize$0$};
  \draw (10, 0.2) -- (10, -0.2) node[below] {\footnotesize$T$};
 \draw[<->] (0, 2.25) -- (1, 2.25) node[above, midway] {\footnotesize$1$};
 \draw[dotted] (1, 0) -- (1, 2.25);
 \draw[<->] (9, 2.25) -- (10, 2.25) node[above, midway] {\footnotesize$1$};
  \draw[dotted] (9, 0) -- (9, 2.25);
  \draw[domain=0:1, smooth, variable=\x, very thick] plot ({\x-2}, {3*(\x*\x)/2)});
  \draw[domain=1:2, smooth, variable=\x, very thick] plot ({\x-2}, {3*(-\x*\x+3*\x-3/2)});
  \draw[domain=2:3, smooth, variable=\x, very thick] plot ({\x-2}, {3*(3-\x)*(3-\x)/2)});
  \draw (-2, 0.2) -- (-2, -0.2) node[below] {\footnotesize$m_1$};
\node at (-2,2)  {$\varphi_{1,m_1}$};
 \draw[domain=0:1, smooth, variable=\x] plot ({\x+3}, {3*(\x*\x)/2)});
  \draw[domain=1:2, smooth, variable=\x] plot ({\x+3}, {3*(-\x*\x+3*\x-3/2)});
  \draw[domain=2:3, smooth, variable=\x] plot ({\x+3}, {3*(3-\x)*(3-\x)/2)});
  \node at (3.25,1.5)  {$\footnotesize \varphi_{1,k}$};
  \draw (3, 0.2) -- (3, -0.2) node[below] {\footnotesize$k$};
\draw[<->] (3, 2.25) -- (6, 2.25) node[above, midway] {\footnotesize$D_1-1$};
 \draw[dotted] (3, 0) -- (3, 2.25);
  \draw[dotted] (6, 0) -- (6, 2.25);
  \draw[domain=0:1, smooth, variable=\x, very thick] plot ({\x+9}, {3*(\x*\x)/2)});
  \draw[domain=1:2, smooth, variable=\x, very thick] plot ({\x+9}, {3*(-\x*\x+3*\x-3/2)});
  \draw[domain=2:3, smooth, variable=\x, very thick] plot ({\x+9}, {3*(3-\x)*(3-\x)/2});
     \draw (9, 0.2) -- (9, -0.2) node[below] {\footnotesize $M_1$};
  \node at (12,2)  {$\varphi_{1,M_1}$};
\end{tikzpicture}} \\
\subfloat[$\Op L_2 = \Op D^2$ ($D_2 = 3$). The $\varphi_{2,k}$ are centered cubic B-splines.]{
\begin{tikzpicture}[scale=0.5]
\path [fill=gray!20] (0,0) rectangle (10,3);
  \draw[->] (-3, 0) -- (13, 0);
  \draw[->] (0, 0) -- (0, 3);
   \draw[<->] (0, 2.25) -- (1, 2.25) node[above, midway] {\footnotesize$1$};
 \draw[dotted] (1, 0) -- (1, 2.25);
 \draw[<->] (9, 2.25) -- (10, 2.25) node[above, midway] {\footnotesize$1$};
  \draw[dotted] (9, 0) -- (9, 2.25);
   \draw (0, 0.2) -- (0, -0.2) node[below] {\footnotesize $0$};
  \draw (10, 0.2) -- (10, -0.2) node[below] {\footnotesize $T$};
  \draw[domain=0:1, smooth, variable=\x, very thick] plot ({\x-3}, {3*(\x*\x*\x)/6)});
  \draw[domain=1:2, smooth, variable=\x, very thick] plot ({\x-3}, {3*(-\x*\x*\x/2+2*\x*\x-2*\x+2/3)});
 \draw[domain=2:3, smooth, variable=\x, very thick] plot ({\x-3}, {3*(\x*\x*\x/2-4*\x*\x+10*\x-22/3)});
   \draw[domain=3:4, smooth, variable=\x, very thick] plot ({\x-3}, {3*(4-\x)*(4-\x)*(4-\x)/6)});
  \draw (-1, 0.2) -- (-1, -0.2) node[below] {\footnotesize$m_2$};
\node at (-2.25,2)  {$\varphi_{2,m_2}$};

  \draw[domain=0:1, smooth, variable=\x] plot ({\x+3}, {3*(\x*\x*\x)/6)});
  \draw[domain=1:2, smooth, variable=\x] plot ({\x+3}, {3*(-\x*\x*\x/2+2*\x*\x-2*\x+2/3)});
 \draw[domain=2:3, smooth, variable=\x] plot ({\x+3}, {3*(\x*\x*\x/2-4*\x*\x+10*\x-22/3)});
   \draw[domain=3:4, smooth, variable=\x] plot ({\x+3}, {3*(4-\x)*(4-\x)*(4-\x)/6)});
     \draw (5, 0.2) -- (5, -0.2) node[below] {\footnotesize$k$};
    \draw[<->] (3, 2.25) -- (5, 2.25) node[above, midway] {\footnotesize$D_2-1$};
     \draw[<->] (5, 2.25) -- (7, 2.25) node[above, midway] {\footnotesize$D_2-1$};
 \draw[dotted] (3, 0) -- (3, 2.25);
  \draw[dotted] (5, 0) -- (5, 2.25);
  \draw[dotted] (7, 0) -- (7, 2.25);
  \node at (3.75,1.5)  {$\footnotesize \varphi_{2,k}$};

  \draw[domain=0:1, smooth, variable=\x, very thick] plot ({\x+9}, {3*(\x*\x*\x)/6)});
  \draw[domain=1:2, smooth, variable=\x, very thick] plot ({\x+9}, {3*(-\x*\x*\x/2+2*\x*\x-2*\x+2/3)});
 \draw[domain=2:3, smooth, variable=\x, very thick] plot ({\x+9}, {3*(\x*\x*\x/2-4*\x*\x+10*\x-22/3)});
   \draw[domain=3:4, smooth, variable=\x, very thick] plot ({\x+9}, {3*(4-\x)*(4-\x)*(4-\x)/6)});
   \draw (11, 0.2) -- (11, -0.2) node[below] {\footnotesize $M_2$};
  \node at (12.25,2)  {$\varphi_{2,M_2}$};
\end{tikzpicture}}
\caption{Examples of boundary basis functions $\varphi_{i,m_i}$ and $\varphi_{i,M_i}$ for $i \in \{1, 2\}$.}
\label{fig:extreme_basis_functions}
\end{figure}

Our choice of basis functions 
together with the assumptions in Section~\ref{sec:assumptions_finite_pb} enable us to restrict Problem~\eqref{eq:discrete_pb} to the interval of interest $I_T$. More precisely, we introduce the indices $m_i, M_i \in \Z$ for $i\in \{1, 2 \}$; the range $[m_i \ldots M_i ]$ corresponds to the set of indices $k$ for which $\mathrm{Supp}(\varphi_{i, k}) \cap I_T \neq \emptyset$, so that the basis function $\varphi_{i, k}$ affects the measurements. Hence, the number of active basis functions (\ie the number of spline coefficients to be optimized) is $N_i = (M_i - m_i - 1)$.  It can easily be verified that we have $m_1 = (-D_1 + 2)$, $M_1 = (T - 1)$, $m_2 = (- D_2 + 2)$, and $M_2 = (T + D_2 - 2)$. See Figure~\ref{fig:extreme_basis_functions} for an illustrative example.

 Finally, we introduce the native digital-filter space
 \begin{align}
\label{eq:native_space_sequence2}
V_2(\Z) =  \Big\{ (c_2[k])_{k \in \Z}: \ \mathrm{Supp}( d_{\Op L_2} \ast c_2) \subset [1 \ldots M_2] \Big\} ,
\end{align}
which is a valid choice because $V_2(\R) \subset \HLtwo$. Indeed, we can verify that, for any $c_2 \in V_2(\Z)$, the function $s_2 = \sum_{k \in \Z} c_2[k] \varphi_{2, k}$ satisfies $\Vert \Op L_2 \{s_2\} \Vert_{L_2}^2 = \Vert \sqrta \ast c_2 \Vert_{\ell_2}^2 = \Vert \sqrtb  \ast (d_{\Op L_2} \ast c_2) \Vert_{\ell_2}^2 < +\infty$, which proves that $s_2 \in \HLtwo$. This is due to the finite support of both $(d_{\Op L_2} \ast c_2)$ and $b^{1/2}$, where the filter $b^{1/2}$ and the decomposition $\sqrta = \sqrtb  \ast d_{\Op L_2}$ are introduced in Proposition~\ref{prop:autocorrelation_B-splines} in Appendix~\ref{app:L_2}.

\begin{remark}
Contrary to $V_1(\Z)$ defined in \eqref{eq:native_space_sequence1}, our choice of $V_2(\Z)$ in \eqref{eq:native_space_sequence2} is not the largest valid space: there exist larger vector spaces such that $V_2(\R) \subset \HLtwo$. However, the support restriction implies that for any $s_2 = \sum_{k \in \Z} c_2[k] \beta_{\Op L_2^\ast \Op L_2}(\cdot - k) \in V_2(\R)$, the function $\Op L_2 \{s_2 \} = \sum_{k \in \Z} (d_{\Op L_2} \ast c_2)[k] \beta_{\Op L_2}^\vee(\cdot - k)$ has a finite support. This is a desirable property both for simplicity of implementation and because it conforms with Theorem~\ref{thm:continuous_RT}, since $s_2^\ast$ in \eqref{eq:continuous_sol_s2} also satisfies this property. Our specific choice of support for $ (d_{\Op L_2} \ast c_2)$ is guided by boundary considerations and will be justified in the proof of Proposition~\ref{prop:finite_pb}. 
\end{remark}

The restriction to a finite number of active spline coefficients leads to finite-dimensional system and regularization matrices. The system matrices are of the form
\begin{align}
\label{eq:system_matrix}
\M{H}_i = \begin{bmatrix}  \V{h}_{m_i} & \cdots & \V{h}_{M_i} \end{bmatrix} \in \R^{M \times N_i}: \quad \V{h}_k = \V\nu(\varphi_{i, k}).
\end{align}
The regularization matrix for the sparse component, denoted by $\M L_1 \in \R^{(N_1 - D_1 + 1) \times N_1}$, is of the form
\begin{align}
\label{eq:regularization_matrix1}
\M{L}_1 = \begin{pmatrix} 
d_{\Op L_1}[D_1-1] & \cdots & d_{\Op L_1}[0] & 0 & \cdots & 0 \\
0 & \ddots & & \ddots & \ddots &\vdots  \\
\vdots & \ddots & \ddots & & \ddots & 0 \\
0 & \cdots & 0 & d_{\Op L_1}[D_1-1]& \cdots & d_{\Op L_1}[0]
\end{pmatrix}.
\end{align}
The second component requires a careful handling of the boundaries in order to achieve exact discretization. This leads to a more complicated expression for the associated regularization matrix, which is given in \eqref{eq:regularization_matrix2} in Appendix~\ref{app:L_2}. 

Finally, we introduce the matrix $\M A \in \R^{N_{0}\times N_1}$ associated to the boundary condition functionals $\V \phi_0$. Our choice of boundary condition functionals $\V \phi_0$ is presented in Appendix~\ref{app:boundary_conditions}. With this choice, the constraint $\V \phi_0(\sum_{k \in \Z} c_1[k] \varphi_{1, k}) = \V 0$ leads to $N_0$ linear constraints on the coefficients $\V c_1 = (c_1[m_1], \ldots , c_1[M_1])$, 
which can be written in matrix form as $\M A \V c_1 = \V 0$. In common cases, these constraints simply lead to the $N_0$ first coefficients of $\V c_1$ to be set to zero, which thus reduces the dimension of the optimization problem.

These matrices enable an exact discretization of Problem~\eqref{eq:discrete_pb}, as shown in Proposition~\ref{prop:finite_pb}, the proof of which being given in Appendix~\ref{app:finite_problem}.

\begin{proposition}[Recasting as a finite problem]
\label{prop:finite_pb}
 Let $\varphi_{1, k} = \beta_{\Op L_1}(\cdot - k)$, $\varphi_{2, k} = \beta_{\Op L_2^\ast \Op L_2}(\cdot - k)$, and let the assumptions 
 in Section~\ref{sec:assumptions_finite_pb} be satisfied.
Then, Problem~\eqref{eq:discrete_pb} is equivalent to the optimization problem
\begin{align}
\label{eq:finite_pb}
 S = \left\{ \argmin_{\substack{(\V c_1, \V c_2) \in \R^{N_1} \times \R^{N_2} \\ \M A \V c_1 = \V 0 }}  J(\V c_1, \V c_2) \right\},
\end{align}
where the cost function is given by 
\begin{align}
J(\V c_1, \V c_2) =& E( \M H_1 \V c_1 + \M H_2 \V c_2, \V y) \\
&+ \lambda_1 \Vert \M L_1 \V c_1 \Vert_1 + \lambda_2 \Vert \M L_2 \V c_2 \Vert_2^2.
\end{align}
The matrices $\M H_i$ and $\M L_i$ for $i \in \{1, 2 \}$ are defined in \eqref{eq:system_matrix}, \eqref{eq:regularization_matrix1}, and \eqref{eq:regularization_matrix2}. This equivalence holds in the sense that there exists a bijective linear mapping from $\Spc{S}_\mathrm{d}$ to $S$
\begin{align}
(c_1, c_2) \mapsto \Big( (c_{1}[m_1], \ldots , c_{1}[M_1]), (c_{2}[m_2], \ldots , c_{2}[M_2] ) \Big)
\end{align}
between their solution sets.
\end{proposition}

The combination of Propositions~\ref{prop:riesz_discretization} and \ref{prop:finite_pb} allows us to solve the continuous-domain infinite-dimensional problem~\eqref{eq:restricted_continuous_pb} by finding a solution $( \opt{\V c}_1, \opt{ \V c}_2) \in S$ of the finite-dimensional problem~\eqref{eq:finite_pb}. We obtain the corresponding solution of~\eqref{eq:restricted_continuous_pb} by extending these vectors to digital filters $(\opt{c}_1, \opt{c}_2) \in \Spc S_\mathrm{d}$ (this extension is unique as specified by Proposition~\ref{prop:finite_pb}), which yields the continuous-domain reconstruction $\opt{s} = \opt{s}_1 + \opt{s}_2$, where $\opt{s}_i = \sum_{k \in \Z} \opt{c}_d[k] \varphi_{i, k}$.

\subsection{Sparsification Step}
\label{sec:sparse_sol}

Although problem~\eqref{eq:finite_pb} can be solved using standard solvers such as the alternating-direction method of multipliers (ADMM), there is no guarantee that such solvers will yield a solution of the desired form specified by Theorem~\ref{thm:continuous_RT}, \ie $\opt{s}_1$ is an $\Op L_1$-spline with fewer than $(M - N_{0,1})$ knots, and $\opt{s}_2$ is a sum of $M$ kernel functions and a null space element. This is a particularly relevant observation for the first component since, at fixed second component $\opt{s}_2$, only extreme-point solutions $\opt{s}_1$ of Problem~\eqref{eq:continuous_pb} take the prescribed form \cite{unser2017splines}. This problem can be alleviated by computing a solution $( \opt{\V c}_1, \opt{\V c}_2)$ to Problem~\eqref{eq:finite_pb}, and then finding an extreme point of the solution set $\V c_1^{\text{extr}} \in \argmin_{\V c_1 \in \R^{N_1}}  J(\V c_1, \opt{\V c}_2)$, which leads to a solution $(\V c_1^{\text{extr}}, \opt{\V c}_2)$ of the prescribed form. This is achieved by recasting the problem as a linear program and using the simplex algorithm \cite{dantzig1955generalized} to reach an extreme-point solution \cite[Theorem 7]{gupta2018continuous}.



\section{Experimental Validation}
\label{sec:applications}

We now validate our reconstruction algorithm in a simulated setting.

\subsection{Experimental Setting}
\subsubsection{Grid Size}
We rescale the problem by a factor $T$ so that the interval of interest $I_T$ is mapped into $[0, 1]$. We tune the finesse of the grid (and the dimension of the optimization task) by varying $T$, which amounts to varying the grid size $h = 1/T$ in the rescaled problem.

\subsubsection{Ground Truth}
\label{sec:ground_truth}
 We generate a ground-truth signal $s^{\mathrm{GT}} = s_1^{\mathrm{GT}} + s_2^{\mathrm{GT}}$. The sparse component  $s_1^{\mathrm{GT}}$ is chosen to be an $\Op L_1$-spline of the form \eqref{eq:spline_def_green} with few jumps, for which gTV is an adequate choice of regularization, as demonstrated by \eqref{eq:continuous_sol_s1} in our representer theorem. For the smooth component $s_2^{\mathrm{GT}}$, we generate a realization of a solution $s_2$ of the stochastic differential equation $\Op L_2 s_2 = w$, where $w$ is a Gaussian white noise with standard deviation $\sigma_2$ by following the method of \cite{dadi2020generating}. The operator $\Op L_2$ then acts as a whitening operator for the stochastic process $s_2$. The reason for this choice is the connection between the minimum mean-square estimation of such stochastic processes and the solutions to variational problems with gTikhonov regularization $\Vert \Op L_2 s_2 \Vert_{L_2}^2$ \cite{wahba1990spline, unser2005generalized, badoual2018periodic}.
 
 \subsubsection{Forward Operator}

Our forward model is the Fourier-domain cosine sampling operator of the form $\nu_1(s) = \int_{0}^1 s(t) \mathrm{d}t$ (DC term) and
\begin{align}
\label{eq:forward_cos}
\nu_m(s) = \int_{0}^1 \cos(\omega_m t + \theta_m) s(t) \mathrm{d}t
\end{align}
for $2 \leq m \leq M$, where the sampling pulsations $\omega_m$ are chosen at random within the interval $(0, \omega_{\max}]$, and the phases $\theta_m$ are chosen at random within the interval $[0, 2 \pi)$. Notice that $\nu_m$ is a Fourier-domain measurement of the restriction of $s$ to the interval of interest $[0,1]$, in conformity with the finite-domain assumption in Section~\ref{sec:assumptions_finite_pb}.

For the data-fidelity term, we use the standard quadratic error $E(\V x, \V y) = \frac12 \Vert \V x - \V y \Vert_2^2$ .
 
  \subsection{Comparison with Non-Composite Models}
 
 
 We now validate our new sparse-plus-smooth model against more standard non-composite models. More precisely, for $i \in \{ 1, 2\}$ we solve the regularized problems
 \begin{align}
     \argmin_{f\in \Spc X_i} \left\{ E(\V \nu (f), \V y) + \lambda \Spc R_i(f) \right\}
 \end{align}
with regularizers $R_1(f) = \Vert \Op L_1 \{ f\} \Vert_{\Spc M}$ (sparse model with native space $\Spc X_1 = \Spc M_{\Op L_1}(\R)$) and $R_2(f) = \Vert \Op L_2 \{ f \} \Vert_{L_2}$ (smooth model with native space $\Spc X_2 = \Spc H_{\Op L_2}(\R)$). We discretize these problems using the reconstruction spaces $V_i(\R)$ described in this paper (without restricting $V_1(\R)$ with the boundary conditions $\V \phi_0$). The sparse model thus amounts to an $\ell_1$-regularized discrete problem which we solve using ADMM, while the smooth model has a closed-form solution that can be obtained by inverting a matrix.

For this comparison, we choose regularization operators $\Op L_1 = \Op D$ and $\Op L_2 = \Op D^2$ with $M=50$ Fourier-domain measurements (cosine sampling with $\omega_{\max} = 100$). We generate the ground-truth signal according to Section~\ref{sec:ground_truth}, with $K_1 = 5$ jumps whose i.i.d. Gaussian amplitudes have the variance $\sigma_1^2 = 1$ for $s_1^{\mathrm{GT}}$. For the smooth component $s_2^{\mathrm{GT}}$, we generate a realization of a Gaussian white noise $w$ with the variance $\sigma_2^2 = 100$, such that $\Op L_2 \{ s_2^{\mathrm{GT}} \} = w$. The measurements are corrupted by some i.i.d. Gaussian white noise $\V n \in \R^M$ so that $\V y = \V \nu (s_{\mathrm{GT}}) + \V n$. We set the signal-to-noise ratio (SNR) between $\V \nu (s_{\mathrm{GT}})$ and $\V n$ to be 50 dB. For all models, we use the grid size $h = 1/2^7$. The regularization parameters are selected through a grid search to maximize the SNR of the reconstructed signal with respect to the ground truth.
 
   \begin{figure}
    \centering
    \includegraphics[width=\linewidth]{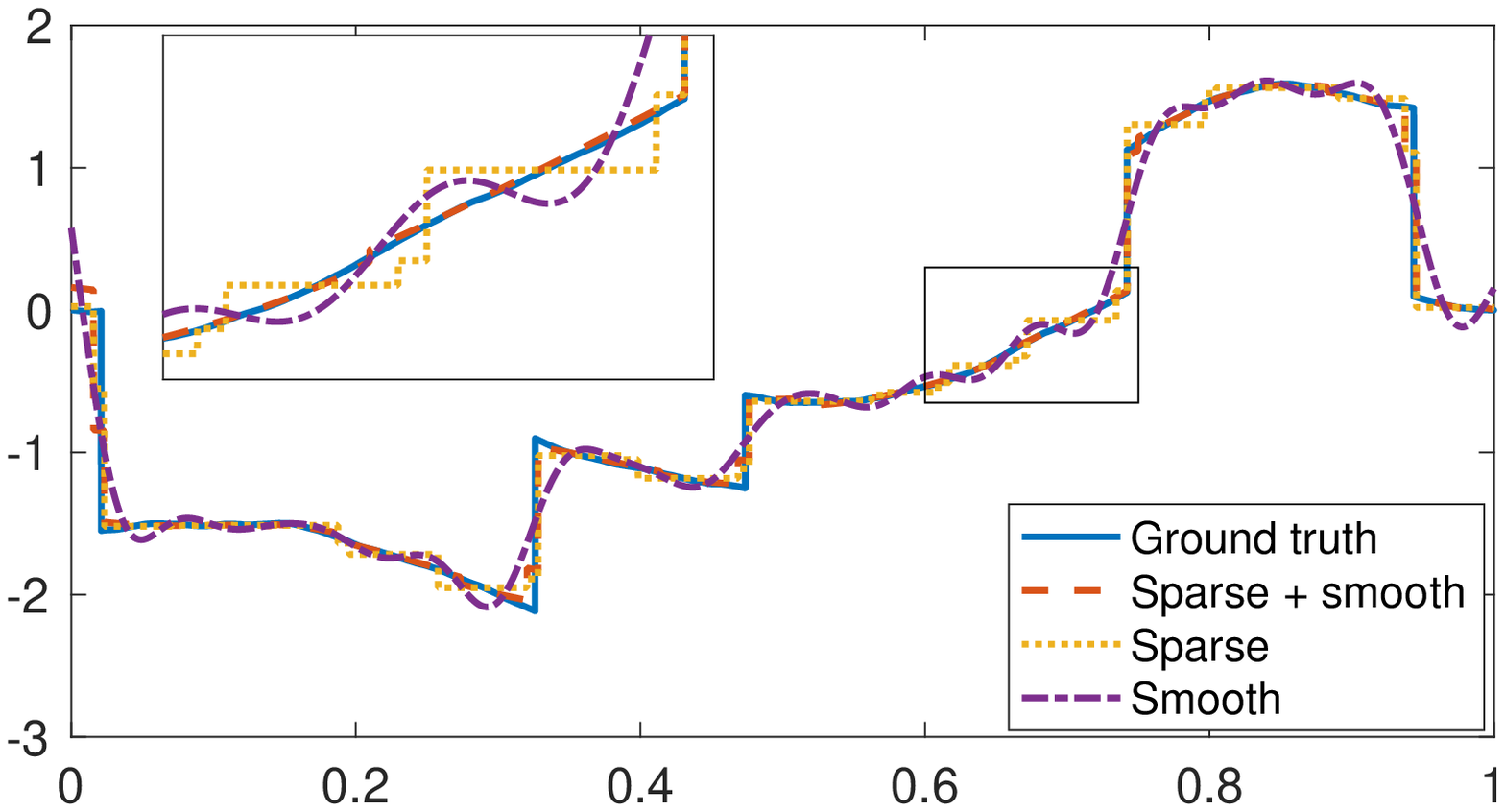}
    \caption{Comparison between our sparse-plus-smooth model and non-composite models with regularization operators $\Op L_1 = \Op D$, $\Op L_2 = \Op D^2$, and $M=50$ Fourier-domain measurements. \\
    Sparse plus smooth : SNR = 21.46 dB with $\lambda_1 = 8 \cdot 10^{-7}$ and $\lambda_2 = 5 \cdot 10^{-10}$. \\
     Sparse : SNR = 21.07 dB with $\lambda = 10^{-9}$. \\
     Smooth : SNR = 18.17 dB with $\lambda = 10^{-11}$.
    }
    \label{fig:comparison_regs}
\end{figure}

 The results of this comparison are shown in Figure~\ref{fig:comparison_regs}. As expected, due to the fact that our sparse-plus-smooth signal model matches the ground truth, our reconstructed signal yields a higher SNR (21.46 dB) than the sparse-only (21.07 dB) and smooth-only (18.17 dB) models. Moreover, our reconstruction is qualitatively much more satisfactory. As can be observed in the zoomed-in section, the sparse-only model is subject to a staircasing phenomenon in the smooth regions of the ground-truth signal, a well-known shortcoming of total-variation regularization. Our reconstruction does not suffer from this phenomenon and is remarkably accurate in the smooth regions. In fact, in our reconstruction, most of the error with respect to the ground truth comes from a lack of precision in the localization of the jumps due to gridding, which is costly in terms of SNR but does not affect much the visual impression. Finally, the smooth-only model fails both visually and in terms of SNR, due to its inability to represent sharp jumps. 

\section{Conclusion}
We have introduced a continuous-domain framework for the reconstruction of multicomponent signals. It assumes two additive components, the first one being sparse and the other being smooth. The reconstruction is performed by solving a regularized inverse problem, using a finite number of measurements of the signal. The form of a solution to this problem is given by our representer theorem. This form justifies the choice of the search space in which we discretize the problem. Our discretization is exact, in the sense that it amounts to solving a continuous-domain optimization problem restricted to our search space. The discretized problem is then solved using our ADMM-based algorithm, which we validate on simulated data.

\appendix

\subsection{Proof of Theorem~\ref{thm:continuous_RT}}
\label{app:proof_RT}

 \textbf{Preliminaries}
 
We extend the biorthogonal system $(\boldsymbol{\phi}_0,\V p_0)$ for $\mathcal{N}_0$ to the biorthogonal systems $(\tilde{\boldsymbol{\phi}}_1,\tilde{\V p}_1)$ and $(\tilde{\boldsymbol{\phi}}_2,\tilde{\V p}_2)$ for  $\mathcal{N}_{\mathrm{L}_1}$ and $ \mathcal{N}_{\mathrm{L}_2} $, respectively, where $\tilde{\boldsymbol{\phi}}_i = \begin{bmatrix}\boldsymbol{\phi}_0 &\boldsymbol{\phi}_i
\end{bmatrix}$ and $\tilde{\V p}_i = \begin{bmatrix}\V p_0 &\V p_i
\end{bmatrix}$ for $i \in \{1,2\}$. 
It is known from \cite[Theorem 4]{unser2019native} that any function $s_1\in \mathcal{M}_{\mathrm{L}_1}(\mathbb{R})$ has the unique decomposition 
\begin{equation}\label{Eq:decomposition1}
s_1 = \mathrm{L}_{1,\tilde{\boldsymbol{\phi}}_1}^{-1} \{ w\} +  \tilde{\mathbf{c}}_{0}^T \V p_0 +  \mathbf{c}_1^T \V p_1,
\end{equation}
 where $w \in \mathcal{M}(\mathbb{R})$, $\tilde{\mathbf{c}}_{0}= \boldsymbol{\phi}_0(s_1) \in\mathbb{R}^{N_0}$, $\mathbf{c}_1= \boldsymbol{\phi}_1(s_1)\in \mathbb{R}^{N_{0, 1}-N_0}$, and $\mathrm{L}_{1,\tilde{\boldsymbol{\phi}_1}}^{-1}$ is the pseudo-inverse operator of $\Op L_1$ for the biorthogonal system $(\tilde{\boldsymbol{\phi}}_1,\tilde{\mathbf{p}}_1)$ \cite[Section 3.2]{unser2019native}. Using this decomposition, we can equip the space $\mathcal{M}_{\mathrm{L}_1}(\mathbb{R})$ with the norm 
 \begin{equation}
     \|s_1\|_{\mathcal{M}_{\rm L_1}} \triangleq \|w\|_{\mathcal{M}} + \|\tilde{\bf c}_0\|_2 + \|{\bf c}_1\|_2.
 \end{equation}
 Finally, an element $s_1\in  \mathcal{M}_{\mathrm{L}_1}(\mathbb{R})$ is in the restricted search space $\Spc M_{\Op L_1, \V \phi_0}(\R)$ if and only if $\tilde{\mathbf{c}}_0= \boldsymbol{0}$. 
 
 Similarly, for any $s_2\in \mathcal{H}_{\mathrm{L}_2}(\mathbb{R})$, there is a unique decomposition
  \begin{equation}\label{Eq:decomposition2}
  s_2 = \mathrm{L}_{2,\tilde{\boldsymbol{\phi}_2}}^{-1} \{ h\} +\mathbf{c}_0^T \V p_0 +\mathbf{c}_2^T \V p_2,
  \end{equation}
   where ${\bf c}_0 =\boldsymbol{\phi}_0(s_2)\in \mathbb{R}^{N_0}$, $\mathbf{c}_2=\boldsymbol{\phi}_2(s_2)\in\mathbb{R}^{N_{0, 2}-N_0}$, and $h\in L_2(\mathbb{R})$. Consequently, the associated norm for the space $\mathcal{H}_{\rm L_2}(\mathbb{R})$ is defined as 
   \begin{equation}
     \|s_2\|_{\mathcal{H}_{\rm L_2}} \triangleq \|h\|_{L_2} + \|{\bf c}_0\|_2 + \|{\bf c}_1\|_2.
 \end{equation}

\textbf{Existence of a Solution} 

The first step is to prove that~\eqref{eq:continuous_pb} has a minimizer. We do so by reformulating the problem as the minimization of a weak*-lower semicontinuous functional over a weak*-compact domain. We then prove the existence by relying on the generalized Weierstrass theorem. 

We denote the cost at the trivial point $(0,0)$ as $\Spc J_0=\mathcal{J}(0,0) = E(\boldsymbol{0},\mathbf{y})$. Adding the constraint $\mathcal{J}(s_1,s_2)\leq \Spc J_0$ does not change the solution set of the original problem, as it must hold for any minimizer of~\eqref{eq:continuous_pb}. So, from now on, we assume that the cost functional is upper-bounded by  $\Spc J_0$. This readily implies that
\begin{align} 
&E(\boldsymbol{\nu}({s}_1+{s}_2),\mathbf{y})\leq \Spc J_0, \label{ineq:datacons}\\
&\|\mathrm{L}_1 \{{s}_1\}\|_{\mathcal{M}} \leq \frac{\Spc J_0}{\lambda_1}, \label{ineq:gtvcons}\\
&\| \mathrm{L}_2\{s_2\} \|_{{L_2}}\leq \sqrt{\frac{\Spc J_0}{\lambda_2}}\label{ineq:Hcons}. 
\end{align}

The coercivity of $E(\cdot,\mathbf{y})$ implies the existence of a constant $C_1>0$ such that
$ E(\mathbf{z} ,\mathbf{y})\leq \Spc J_0 \Rightarrow  \|\mathbf{z}\|_2\leq C_1$. Together with~\eqref{ineq:datacons}, this yields
\begin{equation}
\|\boldsymbol{\nu}({s}_1+{s}_2)\|_2 \leq C_1.    
\end{equation}

Moreover, since $\boldsymbol{\nu}$ is weak*-continuous over $\mathcal{M}_{\rm L_1}(\mathbb{R})$, it is also continuous. This is due to the fact that a Banach space (in this case, the predual of $\mathcal{M}_{\rm L_1}(\mathbb{R})$) is isometrically embedded in its double dual \cite{rudin1986real}. Moreover, by assumption, $\boldsymbol{\nu}$ is continuous over $\mathcal{H}_{\rm L_2}(\mathbb{R})$. Hence, there exists a second constant $C_2>0$ such that 
\begin{equation}\label{Ineq:nucont}
 \|{f}_1\|_{\mathcal{M}_{\mathbf{L}_1}} + \|{f}_2\|_{\mathcal{H}_{L_2}} \leq \frac{\Spc J_0}{\lambda_1} + \sqrt{\frac{\Spc J_0}{\lambda_2}} \Rightarrow \|\boldsymbol{\nu}({f}_1+{f}_2)\|_2 \leq  C_2. 
\end{equation}
Now, by taking
\begin{align}
&f_1 ={s}_1-\boldsymbol{{\phi}}_1(s_1)^T \V p_1, \nonumber \\
&f_2 ={s}_2-\boldsymbol{\phi}_0(s_2)^T \V p_0 -\boldsymbol{\phi}_2(s_2)^T \V p_2, 
\end{align}
and, together with~\eqref{ineq:gtvcons} and~\eqref{ineq:Hcons}, we deduce  that 
\begin{equation}\label{ineq:proj}
\left\|\boldsymbol{\nu}\left({s}_1-\boldsymbol{\phi}_1(s_1)^T{\V p}_1+{s}_2-\boldsymbol{\phi}_0(s_2)^T{\V p}_0-\boldsymbol{\phi}_2(s_2)^T{\V p}_2\right)\right\|_2 \leq  C_2. 
\end{equation}
By using the triangle inequality and the two bounds~\eqref{ineq:proj} and~\eqref{Ineq:nucont}, we have
\begin{equation}
\left\|\boldsymbol{\nu}\left( \boldsymbol{\phi}_1(s_1)^T{\V p}_1+\boldsymbol{\phi}_0(s_2)^T{\V p}_0+\boldsymbol{\phi}_2(s_2)^T{\V p}_2\right)\right\|_2 \leq C_1+C_2.
\end{equation}
Finally, the well-posedness assumption in Theorem~\ref{thm:continuous_RT} ensures the existence of a constant $B>0$ such that 
\begin{equation}\label{Ineq:wellposedness}
\forall q \in \mathcal{N}_{\mathrm{L}_1}+ \mathcal{N}_{\mathrm{L}_2}: B\|\boldsymbol{\phi}_i(q)\|_2\leq \|\boldsymbol{\nu}(q)\|_2, \quad i \in \{0,1,2\}.
\end{equation}
Hence, by taking
\begin{equation}
     q=\boldsymbol{\phi}_1(s_1)^T{\V p}_1+\boldsymbol{\phi}_0(s_2)^T{\V p}_0+\boldsymbol{\phi}_2(s_2)^T{\V p}_2
\end{equation}
 and by applying the Inequality \eqref{Ineq:wellposedness}, we have that
\begin{align}\label{Ineq:NullspaceBound}
\|\boldsymbol{\phi}_1(s_1) \|_{2},\|\boldsymbol{\phi}_0(s_2) \|_2, \|\boldsymbol{\phi}_2(s_2)\|  \leq \frac{ C_1+C_2 }{B}.
\end{align}
 Therefore, the original problem ~\eqref{eq:continuous_pb} is equivalent to the constrained minimization problem
\begin{align}\label{Eq:ConsRegression}
 \min_{\substack{{s}_1\in \MLone \\{s}_2\in \HLtwo}}  \mathcal{J}(s_1,s_2) \quad \text{s.t.} \quad \|{s}_1\|_{\mathcal{M}_{\mathbf{L}_1}} \leq A_1, \| {s}_2 \|_{\mathcal{H}_{\mathrm{L}_2}}\leq  A_2,
\end{align}
where $A_1= \frac{\Spc J_0}{\lambda_1}+ \frac{C_1+C_2}{B}$ and $A_2= \sqrt{\frac{\Spc J_0}{\lambda_2}}+ \frac{C_1+C_2}{B}$. 

The cost functional in~\eqref{Eq:ConsRegression}, which is the same as in~\eqref{eq:continuous_pb}, is weak* lower-semicontinuous. Moreover, the constraint cube is weak*-compact in the product topology due to the Banach-Anaoglu theorem. Hence,~\eqref{Eq:ConsRegression} reaches its infimum, and so does~\eqref{eq:continuous_pb}. 

\textbf{Form of the Solution}
Let $(\tilde{s}_1,\tilde{s}_2)$ be a solution of \eqref{eq:continuous_pb} and consider the minimization problem
\begin{equation}\label{eq:separateL1}
\min_{ {s}_1\in \mathcal{M}_{\mathrm{L}_1,\boldsymbol{\phi}_0 }(\mathbb{R}) }    \|\mathrm{L}_1\{{s}_1\}\|_{\mathcal{M}}  \quad \text{s.t.} \quad \boldsymbol{\nu}({s}_1 ) =\boldsymbol{\nu}(\tilde{s}_1 ).
\end{equation}
Unser \etal have shown in \cite{unser2017splines} that~\eqref{eq:separateL1} has a minimizer $s_1^*$ of the form ~\eqref{eq:continuous_sol_s1}. One can also readily verify that $(s_1^*,\tilde{s}_2)$ is a minimizer of the original problem. Similarly, one can consider the   minimization problem
\begin{equation}\label{eq:separateL2}
\min_{ {s}_2\in \mathcal{H}_{\mathrm{L}_2}(\mathbb{R})  }    \|\mathrm{L}_2\{{s}_2\}\|_{L_2}  \quad \text{s.t.} \quad \boldsymbol{\nu}({s}_2 ) =\boldsymbol{\nu}(\tilde{s}_2).
\end{equation}
It is known from \cite[Theorem 3]{gupta2018continuous} that~\eqref{eq:separateL2} has a minimizer $s_2^*$ of the form~\eqref{eq:continuous_sol_s2}. Again, $(s_1^*,s_2^*)$ is a solution of the original problem, which matches the form specified by Theorem~\ref{thm:continuous_RT}.

\textbf{Uniqueness of the Second Component}
To prove the final statement of Theorem~\ref{thm:continuous_RT}, let us consider two arbitrary pairs of solutions $(\bar{f}_1 ,\bar{f}_2 )$ and $( \tilde{f}_1 ,\tilde{f}_2 )$  of Problem~\eqref{eq:continuous_pb} and let us denote by $\Spc J_{\min}$ their minimal cost value. The convexity of the cost functional yields that, for any $\alpha\in (0,1)$ and $ (f_{\alpha,1},f_{\alpha,2}) = \alpha (\bar{f}_1 ,\bar{f}_2 ) + (1-\alpha) ( \tilde{f}_1 ,\tilde{f}_2 )$, we have 
 \begin{equation}\label{Ineq:cnvx}
 \mathcal{J}(f_{\alpha,1},f_{\alpha,2}) \leq \alpha \mathcal{J}(\bar{f}_1 ,\bar{f}_2 ) + (1-\alpha)\mathcal{J}( \tilde{f}_1 ,\tilde{f}_2 )  = \Spc J_{\min}.
 \end{equation}
 The optimality of $(\bar{f}_1 ,\bar{f}_2 )$ and $( \tilde{f}_1 ,\tilde{f}_2 )$ implies that~\eqref{Ineq:cnvx} must be an equality. In particular, we must have that 
\begin{equation}
    \|{\rm L}_2\{f_{\alpha,2}\}\|_2^2 = \alpha \|{\rm L}_2\{\bar{f}_{2}\}\|_2^2 + (1-\alpha) \|{\rm L}_2\{\tilde{f}_{2}\}\|_2^2.
\end{equation}
Now, due to the strict convexity of  $\|\mathrm{L}_2\{ \cdot \}\|_{L_2}^2$, we deduce that  $\mathrm{L}_2\{  \bar{f}_2 - \tilde{f}_2\}=0$, and hence that $(\bar{f}_2 - \tilde{f}_2) \in \mathcal{N}_{\rm L_2}$. This implies that all solutions have the same second component up to a term in the null space of $\mathrm{L}_2$. 
\subsection{Choice of Boundary Condition Functionals $\V \phi_0$}
\label{app:boundary_conditions}
We discuss here our choice of the boundary-condition functionals $\V \phi_0$ for certain common choices of operators $\Op L_i$. We focus on multiple-order derivative operators $\Op L_i = \Op D^{N_{0, i}}$, although the discussion remains valid for the more general class of rational operators \cite{unser2005cardinalII}, which, to the best of our knowledge, is the largest class of spline-admissible operators that satisfy the first assumption in Section~\ref{sec:assumptions_finite_pb}. The null spaces $\Spc N_{\Op L_i}$ are thus the spaces of polynomials of degree smaller than $N_{0, i}$. We assume for now that we have $N_{0, 1} \leq N_{0, 2}$, in which case we have $\Spc N_{\Op L_1} \subset \Spc N_{\Op L_2}$ and thus $\Spc N_0 = \Spc N_{\Op L_1}$ and $N_0 = (D_1 - 1)$. Then, for any $\epsilon > 0$, the functionals $\phi_0 =  \frac{1}{\epsilon}\mathrm{rect}\left(\frac{\cdot}{\epsilon} \right)$ for $N_0 = 1$ and 
\begin{align}
\label{eq:boundary_canonical}
\V \phi_0 = \left( \delta, \ldots , \delta^{(N_0 - 2)},  \delta^{(N_0 - 1)} \ast \frac{1}{\epsilon}\mathrm{rect}\left(\frac{\cdot}{\epsilon} \right) \right) \ast \delta( \cdot - \frac \epsilon 2),
\end{align}
for $N_0 > 1$, where $\mathrm{rect}(t) = 1$ for $-1/2 \leq t < 1/2$ and 0 elsewhere, are valid choices of a biorthogonal system matched to the basis $\V p_0 = \left( 1, (\cdot), \ldots , \frac{(\cdot)^{N_0-1}}{(N_0 - 1)!} \right) \ast \delta( \cdot - \frac \epsilon 2)$ of $\Spc N_0$. Indeed, one can easily verify that this choice satisfies the biorthonormality relation $\langle \phi_{0, i}, p_{0, j} \rangle = \delta_{i-j}$ (Kronecker delta). Moreover, we have $\phi_{0, i} \in \Spc{X}_{\Op L_1}$ (the predual of $\Spc{M}_{\Op L_1}(\R)$), which implies that $(\V p_0, \V \phi_0)$ is indeed a valid biorthogonal system of $\Spc N_0$ \cite[Proposition 5]{unser2019native}. The fact that $\phi_{0, i} \in \Spc{X}_{\Op L_1}$ is proved in \cite{unser2019representer} for the case $N_0 = 2$; this proof can readily be extended to higher orders. 

The boundary conditions~\eqref{eq:boundary_canonical}, along with a choice of $\epsilon$ such that $\frac \epsilon h$ is arbitrarily small, is numerically equivalent to $\V \phi_0(f) = (f(0), \ldots , f^{(N_0-1)}(0^+))$, where $f^{(N_0-1)}(0^+)$ is the right limit of $f^{(N_0-1)}$ at 0. It can easily be shown in this case that, for $s_1 = \sum_{k \in \Z} c_1[k] \varphi_{1, k} \in V_1(\R)$ with $\varphi_{1, k} = \beta_{\Op L_1}( \cdot - kh)$, we have $\V \phi_0(f) = \V 0 \Leftrightarrow c_1[-N_0 + 1] = \cdots = c_1[0] = 0$, which leads to the constraints $\M A \V c_1 = (c_{1, 1}, \ldots , c_{1, N_0}) = \V 0$ in Problem~\eqref{eq:finite_pb}. This choice simplifies the optimization task by reducing the dimension of the problem, whereas other boundary conditions could lead to more complicated linear constraints and would make the optimization task more difficult.

So far, we have assumed that $N_{0, 1} \leq N_{0, 2}$ since this condition is always satisfied for the most common case $\Op L_1 = \Op D$. However, if $N_{0, 1} > N_{0, 2}$, then it is more convenient to apply the boundary conditions $\V \phi_0$ to the second component, which leads to the same simple boundary conditions $\M A \V c_2 = (c_{2, 1}, \ldots , c_{2, N_0}) = \V 0$. By default, we implicitly consider the more common case $N_{0, 1} \leq N_{0, 2}$ throughout the paper and thus impose the boundary conditions on the first component.

\subsection{Expression of the Regularization Matrix $\M L_2$}
\label{app:L_2}

\textbf{Factorization of the Autocorrelation Filter}

To specify the regularization matrix for the second component, we must first express in a convenient form the autocorrelation filter $\autocorr$ defined in Proposition~\ref{prop:autocorrelation}. This is done in Proposition~\ref{prop:autocorrelation_B-splines}, which gives the expression of $\autocorr$ and its ``square root'' $g$ for the choice of basis function $\varphi_{2} = \beta_{\Op L_2 ^\ast \Op L_2}$ made in Section~\ref{sec:choice_phi2}.
\begin{proposition}[Factorization of the autocorrelation filter]
\label{prop:autocorrelation_B-splines}
Let the assumptions in Section~\ref{sec:assumptions_finite_pb} be satisfied, and let $\varphi_{2} = \beta_{\Op L_2 ^\ast \Op L_2} = \beta_{\Op L_2} \ast \beta_{\Op L_2}^\vee$.
Then, the basis $\{ \varphi_{2, k} \}_{k \in \Z}$ 
forms a Riesz basis as required in Section~\ref{sec:basis_functions}, and the autocorrelation filter $\autocorr$ defined in Proposition~\ref{prop:autocorrelation} is of the form
\begin{align}
\label{eq:autocorrelation_B-spline}
\autocorr = d_{\Op L_2} \ast d_{\Op L_2}^\vee \ast \bseq,
\end{align}
where $\bseq[k] = \beta_{\Op L_2 ^\ast \Op L_2}(k)$ is the B-spline kernel of the operator $\Op L_2 ^\ast \Op L_2$, which is a positive-semidefinite filter supported in $[-(D_2 - 2) \ldots  D_2 - 2 ]$.
The filter $\autocorr$ can thus be factorized as $\autocorr = \sqrta \ast \sqrta^\vee$ with
\begin{align}
\label{eq:sqrt_autocorr}
\sqrta =  d_{\Op L_2} \ast  \sqrtb,
\end{align}
where the filter $\sqrtb$ satisfies $\bseq = \sqrtb \ast (\sqrtb)^\vee$ and is of length $B = (D_2 - 1)$.
\end{proposition}
\begin{proof}
We have that
\begin{align*}
\autocorr [k] &= \langle \Op L_2\{ \varphi_{2, k} \} , \Op L_2\{ \varphi_{2, 0} \}   \rangle_{L_2} \\
&= \langle  \Op L_2^\ast \Op L_2\{ \varphi_{2, k} \} , \varphi_{2, 0}  \rangle_{\Spc H_{\Op L_2}'  \times \Spc H_{\Op L_2}} \\
&=  \langle \sum_{k' \in \Z} d_{\Op L_2^\ast \Op L_2}[k] \delta( \cdot - (k+ k')),  \varphi_{2, 0}  \rangle_{\Spc H_{\Op L_2}' \times \Spc H_{\Op L_2}} \\
&= \sum_{k' \in \Z} d_{\Op L_2^\ast \Op L_2 }[k] \bseq[k + k'] \\
&= (d_{\Op L_2^\ast \Op L_2 } \ast \bseq^\vee)[-k] \\
&=  (d_{\Op L_2^\ast \Op L_2} \ast \bseq)[k],
\end{align*}
where $\langle \cdot, \cdot \rangle_{\Spc H_{\Op L_2}' \times \Spc H_{\Op L_2}}$ denotes the duality product between $\Spc H_{\Op L_2}(\R)$ and its dual $\Spc H_{\Op L_2}'(\R)$, and the last line results from the symmetry of $\autocorr$ and $\bseq$.

Next, we prove that $\bseq$ is positive-semidefinite. Indeed, for any finitely supported filter $c$, we have that
\begin{align}
\sum_{k, k' \in \Z} c[k] c[k'] \bseq[k - k']= \left\Vert \sum_{k \in \Z} c[k] \beta_{\Op L_2}(\cdot - k) \right\Vert_{L_2}^2 \geq 0,
\end{align}
where we have used the property
\begin{align}
\bseq[k] = (\beta_{\Op L_2} \ast \beta_{\Op L_2}^\vee)(k) = \langle \beta_{\Op L_2}, \beta_{\Op L_2}(\cdot - k) \rangle_{L_2}.
\end{align}
Finally, to prove the existence of $\sqrtb$, we notice that $\bseq$ has the finite support $[-(B-1) \ldots B-1]$ due to the finite support $(-D_2, D_2)$ of $\beta_{\Op L_2 ^\ast \Op L_2}$, and we have $B = (D_2 - 1)$. Since $\bseq$ is also symmetric, its $z$-transform satisfies $B(z) = B(z^{-1})$; therefore, for any zero $z_k$ of $B(z)$, $z_k^{-1}$ is also a zero. Moreover, it is well known that $B(\pm 1) \neq 0$, so that zeros must come in pairs $z_k \neq z_k^{-1}$. Hence, $B(z)$ can be written as $B(z) = \prod_{k=1}^B (1 - z_k z) (1 - z_k z^{-1})$. Hence, to take $\sqrtb$ to be the inverse $z$-transform of $B^{1/2}(z) = \prod_{k=1}^B (1 - z_k z^{-1})$ is a valid choice (we clearly have $\bseq = \sqrtb \ast (\sqrtb)^\vee$), and \eqref{eq:sqrt_autocorr} is readily obtained.
\end{proof}

We summarize in Table~\ref{tab:supports} the different filters and their mutual relations. Without loss of generality, we take the filters $d_{\Op L_i}$ for $i \in \{1, 2\}$ to be causal, which leads to causal B-splines. These filters will be useful for the definition of the regularization matrix $\M L_2$.
\begin{table}
\scriptsize
\begin{center}
\begin{tabular}{ C{1.5cm}|C{1.8cm}|C{1.8cm}|C{1.9cm}|C{2.4cm}|C{2.9cm}|C{2.3cm} } 
\hline
\hline
 & $d_{\Op L_1}$ & $d_{\Op L_2}$ & $\autocorr = \bseq \ast d_{\Op L_2} \ast d_{\Op L_2}^\vee$ & $\sqrta = \sqrtb \ast d_{\Op L_2}$ & $\bseq = \left(\beta_{\Op L_2^\ast \Op L_2}(k) \right)_{k \in \Z}$ & $\sqrtb$\\
\hline
\hline
Description & Finite-difference filter for $\Op L_1$ & Finite-difference filter for $\Op L_2$ & Autocorrelation filter for $\Op L_2$ & ``Square root" of $\autocorr$ ($\autocorr = \sqrta \ast \sqrta^\vee$) & Samples of basis function ($\Op L_2^\ast \Op L_2$ B-spline) & ``Square root" of $\bseq$ ($\bseq = \sqrtb \ast (\sqrtb)^\vee$) \\
\hline
Introduced in & Definition~\ref{def:B-spline} & Definition~\ref{def:B-spline} & Proposition~\ref{prop:autocorrelation}  & Proposition~\ref{prop:autocorrelation_B-splines} & Proposition~\ref{prop:autocorrelation_B-splines} & Proposition~\ref{prop:autocorrelation_B-splines} \\
\hline
Support length & $D_1$ & $D_2$ & $2 G - 1 = 4 D_2 - 5 $ & $G = 2 D_2 - 2$ &  $2 B - 1 = 2 D_2 - 3$ &  $B = D_2 - 1 $ \\ 
\hline
Support & $[0 \ldots D_1-1]$ & $[0 \ldots D_2-1]$ & $[-(G-1)$ $ \ldots G-1]$ & $[0 \ldots G-1]$ & $[-(B-1) \ldots B-1]$ & $[0 \ldots B-1]$ \\ 
\hline
\hline
Example $\Op L_2 = \Op D$ & & [1, -1] & [-1, 2, -1] & [1, -1] & [1] & [1] \\
\hline
Example $\Op L_2 = \Op D^2$ & & [1, -2, 1] & $\frac{1}{6}[1, 0, -9, 16$ $, -9, 0, 1]$ & $C[1, \sqrt{3}, {-( 3 + 2\sqrt{3})}$ $, {2 + \sqrt{3}}]$ & $\frac{1}{6} [1, 4, 1]$ & $C [1, 2 + \sqrt{3}]$ \\
\hline
\hline
\end{tabular}
\caption{Relevant filters and their supports ($C = \sqrt{\frac{2 - \sqrt{3}}{6}}$).}
\label{tab:supports}
\end{center}
\end{table}

\textbf{Expression of $\M L_2$}

The regularization matrix $\M L_2 \in \R^{(N_2 - 1) \times N_2}$ for the smooth component is given by
\begin{align}
\label{eq:regularization_matrix2}
\M L_2 = \begin{pmatrix} 
 \M M^- & \vline & & \M 0 &   \\
 \hline 
 & & \M M &  & \\
 \hline
 & \M 0 &  & \vline & \M M^+
\end{pmatrix}.
\end{align}
The central matrix $\M M \in \R ^{(N_2 - G + 1) \times N_2}$ is given by
\begin{align}
\M M = \begin{pmatrix} 
\sqrta[G-1] & \cdots & \sqrta[0] & 0 & \cdots & 0 \\
0 & \ddots & & \ddots & \ddots &\vdots  \\
\vdots & \ddots & \ddots & & \ddots & 0 \\
0 & \cdots & 0 & \sqrta[G-1] & \cdots & \sqrta[0] \\
\end{pmatrix},
\end{align}
where $g$ is defined in Proposition~\ref{prop:autocorrelation_B-splines}. The matrices $\M M^\pm \in \R^{ (B - 1) \times (G - 1)}$ are defined as 
$[\M M^-]_{i,j} = \sqrta^{-(B_2-i)}[G - B + (i-1) - (j-1)]$ and $[\M M^+]_{i,j} = \sqrta^{+i}[G + (i-1) - j)]$  for $1 \leq i \leq (B_2 - 1)$ and $1 \leq j \leq (G - 1)$, where the filter $\sqrta^{\pm k}$ are given by $\sqrta^{-k}  = b^{1/2}\vert_{\{0, \ldots , B-1-k \} } \ast d_{\Op L_2}$ (supported in $[0 \ldots  G-1-k ]$) and $\sqrta^{+k}  = b^{1/2}\vert_{\{k, \ldots, B-1 \} } \ast d_{\Op L_2}$ (supported in $\{k, \ldots  G-1 \}$). Here, the notation $a\vert_{ J}$ refers to the filter $a$ restricted to the set $J$ of indices, with $a\vert_{J}[k] = a[k]$ if $k \in J$, and $a\vert_{J}[k] = 0$ otherwise.

As an illustration, for $\Op L_2 = \Op D$, we have that $B = 1$ and, hence, simply that $\M L_2 = \M M$. For $\Op L_2 = \Op D^2$, we have $B = 2$ and
\begin{align}
\label{eq:regularization_matrix2_D^2}
\M L_2 = \begin{pmatrix} 
 C & -2 C & C & 0 &\cdots &\cdots & 0 \\
 \hline 
 & & & \M M &  & \\
 \hline
0 &\cdots &\cdots & 0 & C' & -2 C' & C'
\end{pmatrix},
\end{align}
where $C = \sqrt{\frac{2 - \sqrt{3}}{6}}$ and $C' = C (2 + \sqrt{3})$.

\subsection{Proof of Proposition~\ref{prop:finite_pb}}
\label{app:finite_problem}

Let $s_i = \sum_{k\in \Z} c_i[k] \varphi_{i, k}$ with $c_i \in V_i(\Z)$ for $i\in \{1, 2\}$. The filters $c_i$ are assumed to have values determined by the vector $\V c_i = (c_i[m_i], \ldots , c_i[M_i])$ at certain points. By definition of $m_i$ and $M_i$, the values of $c_i$ outside these intervals do not affect the measurements $\V \nu(s_i)$, and we clearly have that $\V \nu (s_i) = \M H_i \V c_i$. Therefore, these coefficients solely affect the regularization terms. We now show that, for a solution $(c_1, c_2) \in \Spc S_{\mathrm{d}}$ to problem~\eqref{eq:discrete_pb}, the coefficients are uniquely determined by the vectors $\V c_i$, and that the regularization terms $\Vert d_{\Op L_1} \ast c_1 \Vert_{\ell_1}$ and $\langle c_2, \autocorr \ast c_2 \rangle_{\ell_2}$ can thus be expressed exclusively in terms of these vectors.

Concerning the first component, this is proved in \cite[Proposition 2]{debarre2019b}, which shows that $\Vert \Op L_1\{ s_1 \} \Vert_\Spc{M} = \Vert \M L_1 \V c_1 \Vert_1$. The additional constraint $\M A \V c_1$ comes from $\V \phi_0(s_1) = \V 0$ imposed on the search space $V_1(\Z)$ of Problem~\eqref{eq:discrete_pb}.

We now consider the regularization term for the component $\langle c_2, \autocorr \ast c_2 \rangle_{\ell_2}$. By Proposition~\ref{prop:autocorrelation_B-splines}, we have that $\autocorr = \sqrta \ast \sqrta^\vee$, and, hence, that $\langle c_2, \autocorr \ast c_2 \rangle_{\ell_2} = \langle \sqrta \ast c_2, \sqrta \ast c_2 \rangle_{\ell_2} = \Vert \sqrta \ast c_2 \Vert_{\ell_2}^2$, where $\sqrta = \sqrtb \ast d_{\Op L_2}$. We also have $ \sqrta \ast c_2 =  \sqrtb \ast \autocorr$, where $a = d_{\Op L_2} \ast c_2$ is supported in $[1 \ldots M_2]$ by definition of the native space $V_2(\Z)$ given in \eqref{eq:native_space_sequence2}. Since $a [n] = \sum_{k = 0} ^{D_2 - 1} d_{\Op L_2}[k] c_2[n - k]$, $a [n]$ is entirely determined by the vector $\V c_2$ for $1 \leq n \leq M_2$, which justifies our choice of the space $V_2(\Z)$. For values of $n$ outside this interval, there is a unique way of setting the coefficients $c_2[k]$ in order to nullify $a [n]$ and thus obtain that $c_2 \in V_2(\Z)$. For example, $c_2[M_2 + 1]$ can be set to nullify $a [M_2+1]$ based on the $(D_2-1)$ previous coefficients of $c_2$, and, similarly, all the $c_2[n]$ for $n > M_2 + 1$ can be set recursively to nullify all the $a [k]$ for all $k > M_2+1$. The same argument can be made to show that there is a unique choice $c_2[n]$ for $n < m_2$ that nullifies $a [k]$ for all $k < 1$.

We now compute the values of $(\sqrta \ast c_2)[n]$ in different regimes for $n$. We have that $(\sqrta \ast c_2)[n] = (b^{1/2} \ast a) [n] = \sum_{k=0}^{B_2-1} b^{1/2}[k] a [n-k]$ where $a$ is supported in $[1 \ldots M_2]$. For $B_2 \leq n \leq M_2$, this sum is solely affected by the coefficients $\V c_2 = (c_2[m_2], \ldots , c_2[M_2])$, so that the corresponding terms can be written in matrix form as $\M M \V c$ (the central part of the $\M L_2$ matrix defined in \eqref{eq:regularization_matrix2}). Outside this interval, for example for $n = M_2 + 1$, we have that $(\sqrta \ast c_2)[n] = \sum_{k\in \Z} b^{1/2} \vert_{\{ 1, \ldots , B_2-1\}}[k] a [n-k]$, since the $k=0$ term is anyway nullified by the fact that $a [n] = 0$. An analogous reformulation allows us to have $(\sqrta \ast c_2)[n]$ only depend on the $\V c_2$ coefficients. The same reformulation for all the coefficients $M_2+1 \leq n \leq (M_2 + B_2-1)$ leads to the matrix $\M M^+$ in \eqref{eq:regularization_matrix2}, while a similar argument for coefficients $(\sqrta \ast c_2)[n]$ with $1 \leq n \leq (B_2-1)$ leads to the matrix $\M M^-$.

We have thus proved that the solutions $(c_1, c_2) \in \Spc S_{\mathrm{d}}$ to Problem \eqref{eq:discrete_pb} are uniquely determined by their coefficients $\V c_i = (c_i[m_i], \ldots , c_i[M_i])$ for $i \in \{1, 2 \}$, and that the regularization terms can be written $\Vert d_{\Op L_1} \ast c_1 \Vert_{\ell_1} = \Vert \M L_1  \V c_1 \Vert_1$ and $\Vert \sqrta \ast c_2 \Vert_{\ell_2}^2 = \Vert \M L_2 \V c_2 \Vert_2^2$. This, together with the fact that $\V \nu(\sum_{k\in \Z} c_i[k] \varphi_{i, k}) = \M H_i \V c_i$, proves that $\Spc J_\mathrm{d}(c_1, c_2) = J(\V c_1, \V c_2)$. Conversely, for any $(\V c_1, \V c_2) \in \R^{N_1} \times \R^{N_2}$, there is a unique extension of these vectors to filters $c_i \in V_i(\R)$ such that $\V c_i = (c_i[m_i], \ldots , c_i[M_i])$ and $\Spc J_\mathrm{d}(c_1, c_2) = J(\V c_1, \V c_2)$. These extensions are explicited in \cite[Proposition 2]{debarre2019b} for $c_1$ and earlier in this proof for $c_2$. This proves the existence of the bijective linear mapping between the solution sets $\Spc S_\mathrm{d}$ and $\Spc S$ specified in Proposition~\ref{prop:finite_pb}.

\bibliographystyle{IEEEtran}
\bibliography{main}

\begin{thebibliography}{10}
\providecommand{\url}[1]{#1}
\csname url@samestyle\endcsname
\providecommand{\newblock}{\relax}
\providecommand{\bibinfo}[2]{#2}
\providecommand{\BIBentrySTDinterwordspacing}{\spaceskip=0pt\relax}
\providecommand{\BIBentryALTinterwordstretchfactor}{4}
\providecommand{\BIBentryALTinterwordspacing}{\spaceskip=\fontdimen2\font plus
\BIBentryALTinterwordstretchfactor\fontdimen3\font minus
  \fontdimen4\font\relax}
\providecommand{\BIBforeignlanguage}[2]{{%
\expandafter\ifx\csname l@#1\endcsname\relax
\typeout{** WARNING: IEEEtran.bst: No hyphenation pattern has been}%
\typeout{** loaded for the language `#1'. Using the pattern for}%
\typeout{** the default language instead.}%
\else
\language=\csname l@#1\endcsname
\fi
#2}}
\providecommand{\BIBdecl}{\relax}
\BIBdecl

\bibitem{tikhonov1963solution}
A.~Tikhonov, ``Solution of incorrectly formulated problems and the
  regularization method,'' \emph{Soviet Mathematics}, vol.~4, pp. 1035--1038,
  1963.

\bibitem{donoho2006compressed}
D.~L. Donoho, ``Compressed sensing,'' \emph{IEEE Transactions on Information
  Theory}, vol.~52, no.~4, pp. 1289--1306, 2006.

\bibitem{candes2006compressive}
E.~Candès, ``Compressive sampling,'' in \emph{Proceedings of the International
  Congress of Mathematicians}, vol.~3.\hskip 1em plus 0.5em minus 0.4em\relax
  Madrid, Spain: European Mathematical Society Publishing House, 2006, pp.
  1433--1452.

\bibitem{eldar2012compressed}
Y.~C. Eldar and G.~Kutyniok, \emph{Compressed Sensing: Theory and
  Applications}.\hskip 1em plus 0.5em minus 0.4em\relax Cambridge University
  Press, 2012.

\bibitem{foucart2013mathematical}
S.~Foucart and H.~Rauhut, \emph{A Mathematical Introduction to Compressive
  Sensing}.\hskip 1em plus 0.5em minus 0.4em\relax Birkh{\"a}user Basel, 2013,
  vol.~1, no.~3.

\bibitem{tibshirani1996regression}
R.~Tibshirani, ``Regression shrinkage and selection via the lasso,''
  \emph{Journal of the Royal Statistical Society: Series B (Methodological)},
  vol.~58, no.~1, pp. 267--288, 1996.

\bibitem{candes2006stable}
E.~Candès, J.~Romberg, and T.~Tao, ``Stable signal recovery from incomplete
  and inaccurate measurements,'' \emph{Communications on {P}ure and {A}pplied
  {M}athematics}, vol.~59, no.~8, pp. 1207--1223, 2006.

\bibitem{unser2016representer}
M.~Unser, J.~Fageot, and H.~Gupta, ``Representer theorems for
  sparsity-promoting $\ell_{1}$ regularization,'' \emph{{IEEE} Transactions on
  Information Theory}, vol.~62, no.~9, pp. 5167--5180, 2016.

\bibitem{hastie2015statistical}
T.~Hastie, R.~Tibshirani, and M.~Wainwright, \emph{Statistical {L}earning with
  {S}parsity}.\hskip 1em plus 0.5em minus 0.4em\relax Chapman and Hall/{CRC},
  2015.

\bibitem{beck2009fast}
A.~Beck and M.~Teboulle, ``A fast iterative shrinkage-thresholding algorithm
  for linear inverse problems,'' \emph{{SIAM} {J}ournal on {I}maging
  {S}ciences}, vol.~2, no.~1, pp. 183--202, 2009.

\bibitem{beck2009fasta}
------, ``Fast gradient-based algorithms for constrained total variation image
  denoising and deblurring problems,'' \emph{{IEEE} Transactions on Image
  Processing}, vol.~18, no.~11, pp. 2419--2434, 2009.

\bibitem{chambolle2010first}
A.~Chambolle and T.~Pock, ``A first-order primal-dual algorithm for convex
  problems with applications to imaging,'' \emph{Journal of Mathematical
  Imaging and Vision}, vol.~40, no.~1, pp. 120--145, 2010.

\bibitem{boyd2011distributed}
S.~Boyd, N.~Parikh, E.~Chu, B.~Peleato, and J.~Eckstein, ``Distributed
  optimization and statistical learning via the alternating direction method of
  multipliers,'' \emph{Foundations and Trends{\textregistered} in Machine
  Learning}, vol.~3, no.~1, pp. 1--122, 2011.

\bibitem{demol2004inverse}
C.~De~Mol and M.~Defrise, ``Inverse imaging with mixed penalties,'' in
  \emph{Proceedings URSI EMTS}, Pisa, Italy, 2004, pp. 798--800.

\bibitem{gholami2013balanced}
A.~Gholami and S.~Hosseini, ``A balanced combination of {T}ikhonov and total
  variation regularizations for reconstruction of piecewise-smooth signals,''
  \emph{Signal Processing}, vol.~93, no.~7, pp. 1945--1960, 2013.

\bibitem{naumova2014minimization}
V.~Naumova and S.~Peter, ``Minimization of multi-penalty functionals by
  alternating iterative thresholding and optimal parameter choices,''
  \emph{Inverse Problems}, vol.~30, no.~12, p. 125003, 2014.

\bibitem{daubechies2016sparsity}
I.~Daubechies, M.~Defrise, and C.~De~Mol, ``Sparsity-enforcing regularisation
  and {ISTA} revisited,'' \emph{Inverse Problems}, vol.~32, no.~10, p. 104001,
  2016.

\bibitem{grasmair2018adaptive}
M.~Grasmair, T.~Klock, and V.~Naumova, ``Adaptive multi-penalty regularization
  based on a generalized {L}asso path,'' \emph{Applied and Computational
  Harmonic Analysis}, vol.~49, no.~1, pp. 30--55, 2018.

\bibitem{debarnot2021learning}
V.~Debarnot, P.~Escande, T.~Mangeat, and P.~Weiss, ``Learning low-dimensional
  models of microscopes,'' \emph{{IEEE} Transactions on Computational Imaging},
  vol.~7, pp. 178--190, 2021.

\bibitem{candes2014towards}
E.~J. Cand{\`e}s and C.~Fernandez-Granda, ``Towards a mathematical theory of
  super-resolution,'' \emph{Communications on Pure and Applied Mathematics},
  vol.~67, no.~6, pp. 906--956, 2014.

\bibitem{unser2017splines}
M.~Unser, J.~Fageot, and J.~Ward, ``Splines are universal solutions of linear
  inverse problems with generalized {TV} regularization,'' \emph{{SIAM}
  Review}, vol.~59, no.~4, pp. 769--793, 2017.

\bibitem{aziznejad2018multi}
S.~Aziznejad and M.~Unser, ``Multi-kernel regression with sparsity
  constraint,'' \emph{arXiv preprint arXiv:1811.00836}, 2018.

\bibitem{wahba1990spline}
G.~Wahba, \emph{Spline Models for Observational Data}.\hskip 1em plus 0.5em
  minus 0.4em\relax Philadelphia, USA: Society for Industrial and Applied
  Mathematics, 1990.

\bibitem{schoelkopf2001generalized}
B.~Schölkopf, R.~Herbrich, and A.~Smola, ``A generalized representer
  theorem,'' in \emph{Lecture Notes in Computer Science}, ser. LNCS,
  D.~Helmbold and R.~Williamson, Eds., vol. 2111, no. 2111,
  Max-Planck-Gesellschaft.\hskip 1em plus 0.5em minus 0.4em\relax Berlin,
  Germany: Springer, 2001, pp. 416--426.

\bibitem{gupta2018continuous}
H.~Gupta, J.~Fageot, and M.~Unser, ``Continuous-domain solutions of linear
  inverse problems with {T}ikhonov {\textit{versus}} generalized {TV}
  regularization,'' \emph{{IEEE} Transactions on Signal Processing}, vol.~66,
  no.~17, pp. 4670--4684, 2018.

\bibitem{fisher1975spline}
S.~Fisher and J.~Jerome, ``Spline solutions to ${L}^1$ extremal problems in one
  and several variables,'' \emph{Journal of Approximation Theory}, vol.~13,
  no.~1, pp. 73--83, 1975.

\bibitem{boyer2019representer}
C.~Boyer, A.~Chambolle, Y.~De~Castro, V.~Duval, F.~de~Gournay, and P.~Weiss,
  ``On representer theorems and convex regularization,'' \emph{{SIAM} Journal
  on Optimization}, vol.~29, no.~2, pp. 1260--1281, 2019.

\bibitem{bredies2019sparsity}
K.~Bredies and M.~Carioni, ``Sparsity of solutions for variational inverse
  problems with finite-dimensional data,'' \emph{Calculus of Variations and
  Partial Differential Equations}, vol.~59, no.~1, pp. 1--26, 2019.

\bibitem{fageot2020tv}
J.~Fageot and M.~Simeoni, ``{TV}-based reconstruction of periodic functions,''
  \emph{Inverse Problems}, vol.~36, no.~11, p. 115015, 2020.

\bibitem{debarre2019b}
T.~Debarre, J.~Fageot, H.~Gupta, and M.~Unser, ``B-spline-based exact
  discretization of continuous-domain inverse problems with generalized {TV}
  regularization,'' \emph{{IEEE} Transactions on Information Theory}, vol.~65,
  no.~7, pp. 4457--4470, 2019.

\bibitem{flinth2019exact}
A.~Flinth and P.~Weiss, ``Exact solutions of infinite dimensional
  total-variation regularized problems,'' \emph{Information and Inference: A
  Journal of the {IMA}}, vol.~8, no.~3, pp. 407--443, 2019.

\bibitem{adcock2015generalized}
B.~Adcock and A.~Hansen, ``Generalized sampling and infinite-dimensional
  compressed sensing,'' \emph{Foundations of Computational Mathematics}, pp.
  1--61, 2015.

\bibitem{boor2001practical}
C.~de~Boor, \emph{A {P}ractical {G}uide to {S}plines}.\hskip 1em plus 0.5em
  minus 0.4em\relax Springer-Verlag GmbH, 2001.

\bibitem{unser1993b}
M.~Unser, A.~Aldroubi, and M.~Eden, ``\mbox{{B}-{S}pline} signal processing:
  {P}art {I}---{T}heory,'' \emph{{IEEE} Transactions on Signal Processing},
  vol.~41, no.~2, pp. 821--833, 1993, {IEEE-SPS} best paper award.

\bibitem{unser1999splines}
M.~Unser, ``Splines: {A} perfect fit for signal and image processing,''
  \emph{{IEEE} Signal Processing Magazine}, vol.~16, no.~6, pp. 22--38, 1999.

\bibitem{schoenberg1973cardinal}
I.~Schoenberg, \emph{Cardinal {S}pline {I}nterpolation}.\hskip 1em plus 0.5em
  minus 0.4em\relax Philadelphia, PA: SIAM, 1973.

\bibitem{belge2002efficient}
M.~Belge, M.~Kilmer, and E.~Miller, ``Efficient determination of multiple
  regularization parameters in a generalized l-curve framework,'' \emph{Inverse
  Problems}, vol.~18, no.~4, pp. 1161--1183, 2002.

\bibitem{roth2005fields}
S.~Roth and M.~Black, ``Fields of experts: A framework for learning image
  priors,'' in \emph{2005 {IEEE} Computer Society Conference on Computer Vision
  and Pattern Recognition ({CVPR}{\textquotesingle}05)}, vol.~2.\hskip 1em plus
  0.5em minus 0.4em\relax San Diego, California: {IEEE}, June 20-26 2005, pp.
  860--867.

\bibitem{chen2008multi}
Z.~Chen, Y.~Lu, Y.~Xu, and H.~Yang, ``Multi-parameter {T}ikhonov regularization
  for linear ill-posed operator equations,'' \emph{Journal of Computational
  Mathematics}, vol.~26, no.~1, pp. 37--55, 2008.

\bibitem{lu2010multi}
S.~Lu and S.~Pereverzev, ``Multi-parameter regularization and its numerical
  realization,'' \emph{Numerische Mathematik}, vol. 118, no.~1, pp. 1--31,
  2010.

\bibitem{wang2012multi}
Z.~Wang, ``Multi-parameter {T}ikhonov regularization and model function
  approach to the damped {M}orozov principle for choosing regularization
  parameters,'' \emph{Journal of Computational and Applied Mathematics}, vol.
  236, no.~7, pp. 1815--1832, 2012.

\bibitem{abhishake2016multi}
R.~Abhishake and S.~Sivananthan, ``Multi-penalty regularization in learning
  theory,'' \emph{Journal of {C}omplexity}, vol.~36, pp. 141--165, 2016.

\bibitem{zou2005regularization}
H.~Zou and T.~Hastie, ``Regularization and variable selection via the elastic
  net,'' \emph{Journal of the Royal Statistical Society: Series B (Statistical
  Methodology)}, vol.~67, no.~2, pp. 301--320, 2005.

\bibitem{meyer2001oscillating}
Y.~Meyer, \emph{Oscillating {P}atterns in {I}mage {P}rocessing and {N}onlinear
  {E}volution {E}quations: {T}he {F}ifteenth {D}ean {J}acqueline {B}. {L}ewis
  {M}emorial {L}ectures}.\hskip 1em plus 0.5em minus 0.4em\relax American
  Mathematical Society, 2001, vol.~22.

\bibitem{vese2003modeling}
L.~Vese and S.~Osher, ``Modeling textures with total variation minimization and
  oscillating patterns in image processing,'' \emph{Journal of {S}cientific
  {C}omputing}, vol.~19, no. 1/3, pp. 553--572, 2003.

\bibitem{vese2004image}
------, ``Image denoising and decomposition with total variation minimization
  and oscillatory functions,'' \emph{Journal of Mathematical Imaging and
  Vision}, vol.~20, no. 1/2, pp. 7--18, 2004.

\bibitem{mumford1989optimal}
D.~Mumford and J.~Shah, ``Optimal approximations by piecewise smooth functions
  and associated variational problems,'' \emph{Communications on {P}ure and
  {A}pplied {M}athematics}, vol.~42, no.~5, pp. 577--685, 1989.

\bibitem{schwartz1951theorie}
L.~Schwartz, \emph{Th{\'e}orie des distributions}.\hskip 1em plus 0.5em minus
  0.4em\relax Hermann Paris, 1951, vol.~2.

\bibitem{rudin1986real}
W.~Rudin, \emph{Real and {C}omplex {A}nalysis}.\hskip 1em plus 0.5em minus
  0.4em\relax McGraw-Hill Education, 1986.

\bibitem{unser2019native}
M.~Unser and J.~Fageot, ``Native {B}anach spaces for splines and variational
  inverse problems,'' \emph{arXiv preprint arXiv:1904.10818}, 2019.

\bibitem{daubechies1992ten}
I.~Daubechies, \emph{Ten Lectures on Wavelets}.\hskip 1em plus 0.5em minus
  0.4em\relax Society for Industrial and Applied Mathematics, 1992.

\bibitem{amini2018universal}
A.~Amini, R.~Madani, and M.~Unser, ``A universal formula for generalized
  cardinal \mbox{{B}-Splines},'' \emph{Applied and Computational Harmonic
  Analysis}, vol.~45, no.~2, pp. 341--358, 2018.

\bibitem{debarre2019hybrid}
T.~Debarre, S.~Aziznejad, and M.~Unser, ``Hybrid-spline dictionaries for
  continuous-domain inverse problems,'' \emph{{IEEE} Transactions on Signal
  Processing}, vol.~67, no.~22, pp. 5824--5836, 2019.

\bibitem{bohra2020computation}
P.~Bohra and M.~Unser, ``Computation of ``best'' interpolants in the $l_{p}$
  sense,'' in \emph{Proceedings of the Forty-Fifth IEEE International
  Conference on Acoustics, Speech, and Signal Processing ({ICASSP'20})},
  Barcelona, Kingdom of Spain, May 4-8, 2020, pp. 5505--5509.

\bibitem{unser2005cardinal}
M.~Unser and T.~Blu, ``Cardinal exponential splines: {P}art {I}---{T}heory and
  filtering algorithms,'' \emph{{IEEE} Transactions on Signal Processing},
  vol.~53, no.~4, pp. 1425--1438, 2005.

\bibitem{unser2005cardinalII}
M.~Unser, ``Cardinal exponential splines: {P}art {II}---{T}hink analog, act
  digital,'' \emph{{IEEE} Transactions on Signal Processing}, vol.~53, no.~4,
  pp. 1439--1449, 2005.

\bibitem{dantzig1955generalized}
G.~Dantzig, A.~Orden, and P.~Wolfe, ``The generalized simplex method for
  minimizing a linear form under linear inequality restraints,'' \emph{Pacific
  Journal of Mathematics}, vol.~5, no.~2, pp. 183--195, 1955.

\bibitem{dadi2020generating}
L.~Dadi, S.~Aziznejad, and M.~Unser, ``Generating sparse stochastic processes
  using matched splines,'' \emph{{IEEE} Transactions on Signal Processing},
  vol.~68, pp. 4397--4406, 2020.

\bibitem{unser2005generalized}
M.~Unser and T.~Blu, ``Generalized smoothing splines and the optimal
  discretization of the {W}iener filter,'' \emph{{IEEE} Transactions on Signal
  Processing}, vol.~53, no.~6, pp. 2146--2159, 2005.

\bibitem{badoual2018periodic}
A.~Badoual, J.~Fageot, and M.~Unser, ``Periodic splines and {G}aussian
  processes for the resolution of linear inverse problems,'' \emph{{IEEE}
  Transactions on Signal Processing}, vol.~66, no.~22, pp. 6047--6061, 2018.

\bibitem{unser2019representer}
M.~Unser, ``A representer theorem for deep neural networks,'' \emph{Journal of
  Machine Learning Research}, vol.~20, no. 110, pp. 1--30, 2019.

\end{thebibliography}

\end{document}